\documentclass[letterpaper,11pt]{article}

\usepackage[margin=1in]{geometry}
\usepackage{authblk}

\usepackage{amsmath,amssymb,amsthm}
\usepackage{tabularx}
\usepackage[algoruled,lined,algonl,procnumbered,nosemicolon]{algorithm2e}
\usepackage{lineno}
\usepackage{subcaption}
\usepackage{float}
\usepackage{enumitem}
\usepackage[table]{xcolor}
\usepackage[hyphens]{url}
\usepackage{graphicx,hyperref}
\usepackage{breakcites}
\usepackage{tikz}
\usepackage{booktabs}
\usepackage[textsize=scriptsize]{todonotes}
\usepackage{lmodern}

\usepackage{multirow}

\usepackage{orcidlink}
\usepackage{placeins}  
\usepackage{breakcites}

\newtheorem{theorem}{Theorem}
\newtheorem{corollary}{Corollary}

\newtheorem{lemma}{Lemma}
\newtheorem{observation}{Observation}
\newtheorem{remark}{Remark}

\newtheorem{definition}{Definition}
\newtheorem{example}{Example}

\newtheorem{claim}{Claim}
\newtheorem*{claim*}{Claim}
\newenvironment{claimproof}{\noindent{\em Proof of the claim.}}{\hfill $\diamond$ \par\medskip}

\newcommand{\PWS}{\textnormal{\textsc{Pinwheel Scheduling}}}
\newcommand{\TPWS}{\textnormal{\textsc{Threshold Pinwheel Scheduling}}}
\newcommand{\MVPWS}{\textnormal{\textsc{Maximum-Visit Pinwheel Scheduling}}}

\newcommand{\PM}{\textnormal{\textsc{Position Matching}}}
\newcommand{\tP}{\textnormal{3\textsc{-Partition}}}
\newcommand{\NTDM}{\textnormal{\textsc{Numerical}\ 3\textsc{-Dimensional Matching}}}
\newcommand{\NTDMshort}{\textnormal{\textsc{N3DM}}}

\newcommand{\INTDM}{\textnormal{\textsc{Inequality Numerical}\ 3\textsc{-Dimensional Matching}}}
\newcommand{\INTDMshort}{\textnormal{\textsc{IN3DM}}}

\newcommand{\RNTDM}{\textnormal{\textsc{Restricted Numerical}\ 3\textsc{-Dimensional Matching}}}
\newcommand{\RNTDMshort}{\textnormal{\textsc{RN3DM}}}

\newcommand{\kV}{\textnormal{\ensuremath{k}\textsc{-Visits}}}

\newcommand{\varkV}{\textnormal{\textsc{Var-}\ensuremath{k}\textsc{-Visits}}}

\newcommand{\oneV}{\textnormal{1\textsc{-Visit}}}
\newcommand{\twoV}{\textnormal{2\textsc{-Visits}}}

\def\bigO{\ensuremath{\mathcal{O}}\xspace}

\makeatother

\newcommand{\equalcontribsymbol}{*}

\makeatletter
\newcommand{\equalcontribfootnote}{
  \renewcommand{\thefootnote}{\equalcontribsymbol}%
  \footnotetext{Equal contribution}%
  \renewcommand{\thefootnote}{\arabic{footnote}}%
}
\makeatother

\begin{document}
\title{Finite Pinwheel Scheduling: the k-Visits Problem}

%Pinwheel Scheduling with Finite Repetitions: The Complexity of the k-Visits Problem
%On the k-Visits Scheduling Problem: A Finite Perspective on Pinwheel Complexity

% \author{ 
% Sotiris Kanellopoulos\thanks{Equal contribution} \fnmsep \inst{1, 2}\orcidlink{0009-0006-2999-0580}}

\author[\equalcontribsymbol,1,2]{Sotiris Kanellopoulos \orcidlink{0009-0006-2999-0580} }
\author[\equalcontribsymbol,1,2]{Christos Pergaminelis \orcidlink{0009-0009-8981-3676} }
\author[3]{Maria Kokkou \orcidlink{0009-0009-8892-3494} } 

\author[2,4]{\\ Euripides Markou }
\author[1,2]{Aris Pagourtzis \orcidlink{0000-0002-6220-3722} }
\affil[1]{National Technical University of Athens, Greece}
\affil[2]{Archimedes, Athena Research Center, Greece}
\affil[3]{Paderborn University, Paderborn, Germany}
\affil[4]{University of Ioannina, Greece}

\affil[ ]{\texttt{\{s.kanellopoulos,chr.pergaminelis\}@athenarc.gr, maria.kokkou@uni-paderborn.de, emarkou@cse.uoi.gr, pagour@cs.ntua.gr}}

\date{}

\maketitle
\equalcontribfootnote
\begin{abstract}

\textsc{Pinwheel Scheduling} is a fundamental scheduling problem, in which each task $i$ is associated with a positive integer \emph{deadline} $d_i$, and the objective is to schedule one task per time slot, ensuring each task $i$ perpetually appears at least once in every $d_i$ time slots. Although conjectured to be PSPACE-complete, the complexity of \textsc{Pinwheel Scheduling} remains open. 

We introduce \textsc{$k$-Visits}, a finite version of \textsc{Pinwheel Scheduling}, where given the deadlines of $n$ tasks, the goal is to schedule each task exactly $k$ times. While we observe that the \textsc{$1$-Visit} problem is trivial, we prove that \textsc{$2$-Visits} is strongly NP-complete through a reduction from \textsc{Numerical 3-Dimensional Matching}. %As intermediate steps in the reduction, we define NP-complete variants of \textsc{N$3$DM} which may be of independent interest.
We further extend our strong NP-hardness result to a generalization of \textsc{$k$-Visits} ($k\geq 2$) in which the deadline of each task may vary throughout the schedule, as well as to a similar generalization of \textsc{Pinwheel Scheduling}, thus making progress towards settling the complexity of the latter.

Additionally, we prove that \textsc{$2$-Visits} can be solved in linear time if all deadlines are distinct, rendering it one of the rare natural problems which exhibit the interesting dichotomy of being in P if their input is a set and NP-complete if the input is a multiset. %We achieve this through a Turing reduction from \textsc{$2$-Visits} to a variation of \textsc{N$3$DM}, which we call \textsc{Position Matching}.
We also present an FPT algorithm for \textsc{$2$-Visits} parameterized by a value related to how close the input deadlines are to each other, as well as a linear-time algorithm for instances with up to two distinct values of deadlines close to each other.

\end{abstract}

\clearpage

\section{Introduction}

Deadline-based scheduling problems are fundamental in both theory and practice, arising in diverse settings such as manufacturing, communications,  satellite operations and are even present in everyday tasks. 
For example, imagine a single delivery truck tasked with supplying several stores, each with diverse restocking needs. The challenge is to figure out a schedule that ensures all stores are resupplied before running out of stock. How can the truck manage this, given the varying demands of the stores?

The \textsc{Pinwheel Scheduling} problem~\cite{Holte_Pinwheel} (also known as \emph{Windows Scheduling} with one channel~\cite{Bar-Noy_Windows,Bar-Noy_0.6}, or \emph{Periodic Scheduling}, e.g.~\cite{Bar-Noy_Periodic}) formalizes such questions in perpetual settings as follows. Given a multiset $D=\{d_1,\ldots,d_n\}$ of positive integers (deadlines), the goal is to determine whether there exists an infinite sequence over $\{1,\ldots,n\}$ such that any subsequence of $d_i$ consecutive entries contains at least one instantiation of~$i$. Despite its simple description, this problem is still associated with many open questions, with its complexity in particular remaining open for over three decades. Although conjectured to be PSPACE-complete, NP-hardness had until recently only been proven when the multiplicity of each deadline is given as input (compact encoding); for the case in which the input is given explicitly as a multiset, NP-hardness was proven very recently by Kleinberg and Mishra~\cite{kleinberg_mishra_NP_hardness}. See Table~\ref{tab:comparison} for a comprehensive list of known complexity results. Note that instances with sum of inverse deadlines (\emph{density}) greater than $1$ do not admit a feasible schedule. In a recent breakthrough, Kawamura [STOC 2024]~\cite{Kawamura_5/6_stoc} also proved that instances with density not exceeding~$5/6$ always admit a feasible schedule, which was a long-standing conjecture proposed by Chan and Chin~\cite{Chan_conjecture}.

%\arrayrulecolor{black!30}\cline{1-2}\arrayrulecolor{black}
\newcommand{\lightrule}{\arrayrulecolor{black!30}\hline\arrayrulecolor{black}}
{
\begin{table}[ht]
\centering
\begin{tabular}{|l|l!{\color{black!30}\vrule}l|}
\hline
 & \PWS & \twoV\ [this work] \\
\hline
Feasible when density $\leq 5/6$      & Always~\cite{Chan_conjecture,Kawamura_5/6_stoc} & Always\\
\hline
Feasible when density $>1$ & Never~\cite{Holte_Pinwheel} & Sometimes \\
\hline

Complexity (explicit input) & In PSPACE~\cite{Holte_Pinwheel} & Strongly NP-complete \\ 
\arrayrulecolor{black!30}\cline{2-3}\arrayrulecolor{black}
& Conjectured PSPACE-complete~\cite{Bosman_Replenishment} & In P for distinct deadlines \\ 
\arrayrulecolor{black!30}\cline{2-2}\arrayrulecolor{black}
& Weakly NP-hard~\cite{kleinberg_mishra_NP_hardness} & \\ 
\arrayrulecolor{black!30}\cline{2-2}\arrayrulecolor{black}
& Unknown whether in NP & \\ 
\arrayrulecolor{black!30}\cline{2-2}\arrayrulecolor{black}
& (Pseudo)poly-time unlikely~\cite{Jacobs_Window_Scheduling_Complexity,Pinwheel_ISAAC_2025}  & \\ 
\arrayrulecolor{black!30}\cline{2-2}\arrayrulecolor{black}
& In NP when Density $=1$ \cite{Holte_Pinwheel} & \\
\hline

Complexity (compact input) & Weakly NP-hard~\cite{Holte_Pinwheel, Bar-Noy_0.6} \footnotemark & Strongly NP-complete  \\
\arrayrulecolor{black!30}\cline{2-2}\arrayrulecolor{black}
& Unknown whether in NP & \\ 
\hline

Tractable special cases &  Two distinct numbers~\cite{Holte_2_distinct} & Two distinct numbers \\
& & (per cluster) \\
\arrayrulecolor{black!30}\cline{3-3}\arrayrulecolor{black}
& %Three distinct numbers~\cite{Lin_3distinct}
& All deadlines distinct \\ 
\arrayrulecolor{black!30}\cline{3-3}\arrayrulecolor{black}
& & Constant cluster size \\
\hline

\end{tabular}
\caption{Our results for $\twoV$ compared to their $\PWS$ counterparts.}
\label{tab:comparison}
\end{table}
}

% \begin{table}[ht]
%     \centering
%     \begin{tabularx}{\textwidth}{lXl}
%     \toprule
%      & $\PWS$ & \centering $\twoV$ [this work] \tabularnewline \midrule
%     Feasible when density $\leq 5/6$                 & Always~\cite{Chan_conjecture,Kawamura_5/6_stoc}             & Always \tabularnewline \midrule
%     Feasible when density $>1$    & Never~\cite{Holte_Pinwheel}                     & Sometimes \tabularnewline \midrule
%     Complexity (explicit input) & In PSPACE~\cite{Holte_Pinwheel} & Strongly NP-complete \\
%      & Conjectured PSPACE-complete~\cite{Bosman_Replenishment} & In P for distinct deadlines \\ & Unknown whether NP-hard \\ & Unknown whether in NP \\ & (Pseudo)poly-time unlikely~\cite{Jacobs_Window_Scheduling_Complexity}  \\ & In NP when Density $=1$ \cite{Holte_Pinwheel} \tabularnewline \midrule 
%     Complexity (compact input) & NP-hard~\cite{Holte_Pinwheel, Bar-Noy_0.6} \footnotemark & Strongly NP-complete  \\
%     & Unknown whether in NP \tabularnewline \midrule
%     Tractable special cases &  Two distinct numbers~\cite{Holte_2_distinct} & Two distinct numbers \\
%     & %Three distinct numbers~\cite{Lin_3distinct}
%     & All deadlines distinct \\ & & Bounded cluster size \\
%     \bottomrule
%     \end{tabularx}
%     \caption{Old.}
%     \label{tab:old}
% \end{table}

%\textcolor{red}{29 may not prove exactly that; seems to only work up to density 5/6 - remove from the table probably}

%\footnotetext[1]{To the best of our knowledge, this claim by Jacobs and Longo~\cite{Jacobs_Window_Scheduling_Complexity} has not appeared in any peer-reviewed venue.}

\footnotetext[1]{The NP-hardness proof of Holte et al.~\cite{Holte_Pinwheel} for compact encoding does not seem to have been published; nevertheless, a proof by Bar-Noy et al. can be found in~\cite{Bar-Noy_0.6}.}

%The $\kV$ problem can be modeled by an agent traversing a complete graph with self-loops where all edges have weight $1$ and each node $i$ has a deadline $d_i$. The goal is to visit each node $i$ a total of $k$ times without its deadline expiring. In each time unit, the agent can traverse an edge and renew the deadline of the node it visits. Equivalently, the problem can be modeled by a star graph where the center is a special node with no deadline, all edges have weight $1/2$ and the agent has to go back to the center after visiting a node. In both representations, we assume that the agent may start from any node. In the context of this paper, we always assume graphs as described above, therefore it is sufficient to only consider the deadlines as input.

In this work, we introduce a finite variant of \textsc{Pinwheel Scheduling}, which we call \textsc{$k$-Visits}. The motivation for this variant is twofold: first, certain deadline-based applications may only require a finite amount of task repetitions; second, the study of this variant may be useful as an intermediate step for settling the complexity of \textsc{Pinwheel Scheduling}. Notably, the proof of PSPACE-membership by Holte et al.~\cite{Holte_Pinwheel} implies that the existence of an infinite schedule is equivalent to the existence of a (finite) schedule of length equal to the product of the input deadlines, strengthening the connection between finite and infinite variants of the problem.

We find that the simplest version of \textsc{$k$-Visits} in which $k = 1$ is trivial, but the $k=2$ version turns out to be a challenging combinatorial problem. Hence, the main focus of this work is the $k=2$ version, i.e., the \textsc{$2$-Visits} problem, for which we prove strong NP-completeness and examine tractable special cases. Notably, this strong NP-hardness result transfers to a generalization of \textsc{Pinwheel Scheduling} in which the deadline of each task may vary throughout the schedule. This is of particular interest as, to the best of our knowledge, no such result exists for \textsc{Pinwheel Scheduling}. We summarize the main differences in what is known for the two problems in Table~\ref{tab:comparison}.

We remark that our hardness result may prove useful in the future for settling the PSPACE-completeness of \PWS, since PSPACE-hardness proofs for periodic problems often rely on modifying the NP-hardness proof of some finite version (e.g.,~\cite{Los_Alamos_periodic_PSPACE,Papadimitriou_book}).

%Although \textsc{Pinwheel Scheduling} and \textsc{$2$-Visits} share the common difficulty of needing to schedule tasks multiple times in a way that the deadline of each task is always respected, the results that we prove for our more restricted version of the problem give a clearer picture of what can and cannot be scheduled. 

%\textcolor{red}{Even though pspace-complete is conjectured, it is not known to be np-hard for non-compact input. - must check the randomized exp hypo paper mentioned by Bosman to be sure}

%Our results establish new theoretical bounds and offer a clearer understanding of the boundary between tractability and intractability in deadline-based scheduling.

\subsection{Related Work}

Holte, Mok, Rosier, Tulchinsky and Varvel~\cite{Holte_Pinwheel} introduced the \textsc{Pinwheel Scheduling} problem in 1989, proving that it is contained in PSPACE by showing that an infinite schedule exists if and only if an infinite schedule with finite period exists. In the same paper, the authors introduce the concept of \emph{density} and prove membership in NP when restricted to instances with density equal to~$1$, while also claiming NP-hardness for a compact input (see~\cite{Bar-Noy_0.6} for a proof). Since the problem's introduction, it has remained open whether it is contained in NP or is NP-hard (for explicit input), despite the conjecture of PSPACE-completeness~\cite{Bosman_Replenishment}. Regardless, it has been used to classify the complexity of certain inventory routing problems~\cite{inventory_routing}, thus further enhancing the importance of the aforementioned open questions. Complexity results for \PWS\ include~\cite{Jacobs_Window_Scheduling_Complexity,Pinwheel_ISAAC_2025}, which show that the problem admits no (pseudo)polynomial-time algorithms under standard complexity assumptions, and~\cite{PSPACE_UAV}, which shows PSPACE-completeness for a weighted generalization of the problem. After the initial version of this paper was published, a breakthrough paper by Kleinberg and Mishra~\cite{kleinberg_mishra_NP_hardness} finally showed weak NP-hardness for \textsc{Pinwheel Scheduling}.

Since \textsc{Pinwheel Scheduling} has been conjectured to be hard, there have been various attempts to determine easy special cases. Variants with constant amounts of distinct deadlines have been studied in~\cite{Holte_Pinwheel,Holte_2_distinct,Lin_3distinct}. Additionally, there exists a substantial line of research proving that feasible schedules always exist when the density of the input is bounded by some constant. Holte et al.~\cite{Holte_Pinwheel} gave the first such bound of $0.5$; Bar-Noy, Ladner and Tamir~\cite{Bar-Noy_0.6}, Chan and Chin~\cite{Chan_0.7,Chan_conjecture}, Fishburn and Lagarias~\cite{Fishburn_density} all improved this bound, with~\cite{Chan_conjecture} conjecturing that the bound can be improved up to~$5/6$.\footnote{Note that there exists an instance with density~$5/6 +\varepsilon$ that admits no feasible schedule (cf.~\cite{Gasieniec_towards_5/6}).} Gasieniec, Smith and Wild~\cite{Gasieniec_towards_5/6} proved the $5/6$-conjecture for instances of size at most~$12$. In a recent breakthrough paper, Kawamura~\cite{Kawamura_5/6_stoc} finally proved the conjecture in the general setting.

There is also extensive bibliography for variations and generalizations of \textsc{Pinwheel Scheduling}. For example, \cite{SOFSEM_real_periods} studies a variant with real periods, \cite{Feinberg_Generalized_Pinwheel,SOFSEM_Kusano_pinwheel_durations} a variant in which each task additionally has a duration, and~\cite{Kawamura_pinwheel_cover} a variant in which each task has to be executed
\emph{at most} once in a specified number of time units. \emph{Bamboo Garden Trimming} (BGT)~\cite{Bamboo_first} is a closely related optimization problem (see also~\cite{Bamboo_second,Bamboo_approx_2,Bamboo_approx_1}), in which one bamboo is trimmed at a time in a garden of bamboos growing with different rates and the objective is to minimize the maximum height. The best constant approximation algorithm for BGT is given by Mishra~\cite{Patrolling_SODA_2026}, achieving an approximation ratio of $9/7$. Generalizations of BGT are studied in~\cite{Biktairov_Polyamorous,Schewior_Combinatorial_Perpetual_Scheduling_ICALP}.

In addition to the scheduling-based approach, another line of research addresses the problem through a graph-theoretic perspective, where the objective is to visit every node of a graph with a certain frequency. The most basic form of this problem in which each node must be visited once is known as \emph{Exploration} \cite{shannon1993presentation} %In this context, the focus is typically on determining efficient strategies with respect to memory (e.g., \cite{menc2017time}), number of moves (e.g., \cite{Panaite}) or node capacity for storing information (e.g., \cite{reingold}). Furthermore, the problem is studied both in the context of one (e.g., \cite{cohen2008label}) or multiple computational entities (e.g., \cite{das2007map}) evolving in a static or dynamic (e.g., \cite{di2020distributed}) graph. 
(for a survey see \cite{das2019graph}). An extension of the Exploration problem, where each node must be visited within specific time constraints, is studied in \cite{czyzowicz2018exploring}. In the case where each node of the graph needs to be repeatedly visited, the problem is known as \emph{Perpetual Exploration} (e.g., \cite{blin2010exclusive}) or \emph{Patrolling} (for a recent survey see \cite{Czyzowicz_Patrolling}). %In particular, in the Patrolling problem nodes also often have constraints on the maximum allowable time between successive visits to the same node \textcolor{blue}{(done, next sentence)} or the goal is to minimize the time between any two successive visits of every node.
%In particular, in the Patrolling problem, the objective is often one of the following: either to satisfy constraints on the maximum allowable time between successive visits to the same node or to minimize the time between any two consecutive visits to each node.
The most relevant variation of Patrolling to our work is known as \emph{Patrolling with Unbalanced Frequencies} \cite{puf2018}, in which the objective is to satisfy constraints on the maximum allowable time between successive visits to the same node. However, results in that context focus on continuous lines (e.g., \cite{puf2018,damaschke2020two}) and on multiple agents (e.g., \cite{damaschke2020two}), whereas our setting is analogous to complete graphs and a single agent. 

The finite versions of \PWS\ defined here have been further studied in~\cite{ICALP_kVisits}, which studies density thresholds and establishes additional hardness and tractability results, and~\cite{finite_covering}, which defines analogous finite versions for \textsc{Pinwheel Covering} (cf.~\cite{Kawamura_pinwheel_cover}).

\subsection{Our Contribution}

We introduce a finite variant of the classical \textsc{Pinwheel Scheduling} problem, which we call \textsc{$k$-Visits}. This problem can be modeled by an agent traversing a complete graph with loops, with each node having a deadline and the goal being to visit each node exactly $k$ times without its deadline expiring between consecutive visits. While we observe that \textsc{$1$-Visit} is solvable by a trivial algorithm, we find that \textsc{$2$-Visits} is a challenging combinatorial problem. In particular, our main result establishes that \textsc{$2$-Visits} is strongly NP-complete, even when deadlines are given explicitly. This is in contrast to the current status of \textsc{Pinwheel Scheduling}, which is not known to be strongly NP-hard under any input encoding. Our main result implies strong NP-hardness for a generalization of \textsc{Pinwheel Scheduling} in which the deadline of each task may change once throughout the schedule, as well as for an optimization variant seeking to maximize the number of executions of each task. We thus make significant progress towards understanding the computational complexity of \textsc{Pinwheel Scheduling} and other deadline-based scheduling problems. All hardness results in this paper concern explicit input encoding and, thus, also hold for compact encoding.

\begin{figure}[ht]
    \centering
    \includegraphics[scale=.9]{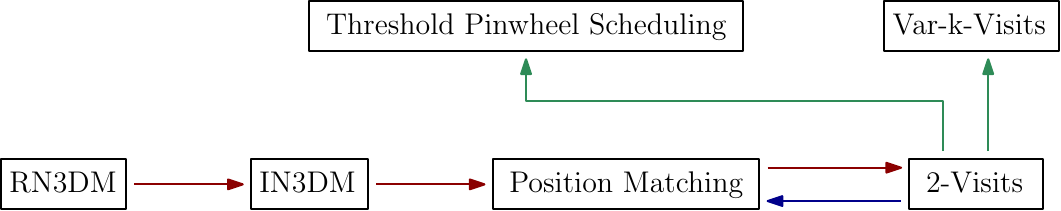}
    \caption{A map of our reductions. $\RNTDMshort$ (Def.~\ref{def:rn3dm}) and $\INTDMshort$ (Def.~\ref{def:numerical_ineq}) are variants of $\NTDM$. Reductions denoted by red (left to right) arrows are used to prove strong NP-completeness for $\twoV$ (Section~\ref{sec:two_visits_hardness}). The blue-arrow (right to left) reduction is a linear time Turing reduction from $\twoV$ to $\PM$, and is the backbone of the algorithms we propose for \textsc{$2$-Visits} (Section~\ref{sec:two_visits_algo}). Lastly, the green-arrow (bottom-up) reductions prove strong NP-hardness for generalizations of \textsc{$k$-Visits} and \textsc{Pinwheel Scheduling} (Section~\ref{sec:hardness_extra}).}
    \label{fig:reductions}
\end{figure}

A more detailed description of our hardness results follows.
We prove strong NP-completeness for \textsc{$2$-Visits} via a surprising chain of reductions starting from a variation of \NTDM\ ($\NTDMshort$) defined in~\cite{Flow_shop_Yu}, and then transfer this hardness result to generalizations of \textsc{$k$-Visits} ($k\geq 2$) and \textsc{Pinwheel Scheduling}. As intermediate steps in our reductions, we define NP-complete variations of $\NTDMshort$ which may be of independent interest. In particular, we introduce a severely restricted variation of $\NTDMshort$, which we call \textsc{Position Matching}, serving as a key component in our results. Our reductions are summarized in Figure \ref{fig:reductions}. A major technical challenge stems from the  constraints of \textsc{Position Matching} and particularly from the concept of \emph{discretized sequences} (see Def.~\ref{def:disc}) which is fundamental for our definitions.

Despite the intractability of \textsc{$2$-Visits}, we identify tractable special cases through a linear-time Turing reduction from \textsc{$2$-Visits} to \textsc{Position Matching}. Our first result is a linear-time algorithm when all deadlines are distinct. This leads to an interesting dichotomy: our hardness result for \textsc{$2$-Visits} holds when the input is a \emph{multiset}, but the problem can be solved in linear time when the input is a simple set. This sensitivity to input multiplicity is rare for natural problems and highlights an interesting structural boundary in the complexity landscape of deadline-based problems. Additionally, we provide an FPT algorithm parameterized by the \emph{maximum cluster size} (see Definitions~\ref{def:disc},~\ref{def:cluster}); intuitively, this is a parameter related to how close the input deadlines are to each other. Lastly, we show a linear-time algorithm when there are at most two distinct deadlines per \emph{cluster} (Def.~\ref{def:cluster}).

Our paper is structured as follows. In Section~\ref{sec:one_visit}, we give a brief overview of the \textsc{$1$-Visit} problem. In Section~\ref{sec:two_visits_algo} we study various properties of \textsc{$2$-Visits} schedules, leading to special-case polynomial-time algorithms for the problem. In Section~\ref{sec:two_visits_hardness}, we prove strong NP-completeness for \textsc{$2$-Visits}, which we consider our main technical contribution. Finally, in Section~\ref{sec:hardness_extra} we extend our hardness result to generalizations of \textsc{$k$-Visits} ($k\geq 2$) and of \textsc{Pinwheel Scheduling}. We stress that Section~\ref{sec:two_visits_algo} (Theorems~\ref{theorem:second_visits} and~\ref{theorem:reduction} in particular) provides intuition for the connection between \textsc{$2$-Visits} and \textsc{Position Matching}, which is necessary to understand our main result in Section~\ref{sec:two_visits_hardness}.

\subsubsection*{Prior Work}

This work is an extension of the paper with the same title presented in SODA 2026~\cite{kVisits}. Compared to the conference version, this paper contains detailed proofs for all statements (when required) and many minor revisions and corrections. In particular, the result in Section~\ref{subsec:two_distinct} has been improved compared to the conference version: we prove that \twoV\ is, in fact, tractable when at most two distinct deadlines correspond to each cluster, instead of demanding that the whole input consists of two distinct deadlines. Additionally, the running time of the FPT algorithm of Section~\ref{subsec:cluster} has been improved. We also expand our results in Section~\ref{sec:hardness_extra}, providing additional proofs and a strong NP-hardness corollary for an optimization version of \PWS.
Lastly, this version includes some related work that we were unaware of or that was unpublished at the time of the SODA 2026 submission.

%We recommend that the reader study Section~\ref{sec:two_visits_algo} before Section~\ref{sec:two_visits_hardness}, since it provides some much-needed intuition for the connection between \textsc{$2$-Visits} and \textsc{Position Matching}.

% \begin{itemize}
% \item (ok) Mention it's rare for a problem to be in P for distinct values and NP-hard for non-distinct values.
% \item (ok) Briefly mention and define density for $\PWS$  \cite{Holte_Pinwheel}
% \item (ok) Mention that $\kV$ instances with density larger than $1$ may admit feasible schedules (in contrast to $\PWS$).
% \item (ok) Add a figure showing the reductions we do
% \item First strong np-hardness result for pinwheel variation (?) - must check bibliography - Bosman has a lot of history to check
% \item (ok) Clearly state that all our hardness results are for non-compact representation (and thus also for compact, which makes the problem harder).
% \end{itemize}

\section{Preliminaries and problem definition}

We use the notations $[n]=\{1,\dots, n\}$ and $[m,n] = \{m,\ldots,n\}$ for positive integers $m\leq n$. For all problems studied in this work, we assume that all input (multi)sets are given in a sorted (non-decreasing) order. When reducing from problem $A$ to problem $B$ in the context of this paper, we restrict $A$ to NP-hard instances, if such a class of instances is known; note that, if $A\in \textsc{NP}$, then all instances of $A$ reduce to this restricted class in polynomial time, by the definition of NP-hardness. Throughout the paper, all running-time bounds are given in the word RAM model, in which standard operations on word-sized integers are assumed to take constant time.

\subsection{Preliminaries}

\begin{definition}[$\PWS$ \cite{Holte_Pinwheel}]
    Given a non-decreasing sequence of $n$ positive integers (deadlines) $d_i$, $i\in [n]$, the $\PWS$ problem asks whether there exists an infinite schedule $p_1, p_2, \ldots$ , where $p_j \in [n]$ for $j\in \mathbb{N}$, such that for all $i\in [n]$ any $d_i$ consecutive entries contain at least one occurrence of $i$.
\end{definition}

\begin{definition}[Density]
    For a non-decreasing sequence of $n$ positive integers (deadlines) $d_i$, $i\in [n]$, we define its \emph{density} as $\sum_{i=1}^n 1/d_i$.
\end{definition}

%\textcolor{red}{Mention hardness results/conjectures for PW here?}

\begin{definition}[$\NTDMshort$]
    Given three multisets of positive integers $A=\{a_1,\ldots,a_n\}$, $B=\{b_1,\ldots,b_n\}$, $C=\{c_1,\ldots,c_n\}$ and an integer $\sigma$ such that
    $$\sum_{i=1}^n (a_i+b_i+c_i) = n\sigma,$$
    the $\NTDM$ $(\NTDMshort)$ problem asks whether there is a subset $M$ of $A \times B \times C$ s.t. every $a_i \in A$, $b_i \in B$, $c_i \in C$ occurs exactly once in $M$ and for every triplet $(a,b,c)\in M$ it holds that $a+b+c=\sigma$.
\end{definition}

$\NTDMshort$ is a variation of the classic strongly NP-complete problem $\tP$ and is also strongly NP-complete~\cite{garey2002computers}. We will use variations of $\NTDMshort$ for our reductions.

\subsection{The k-Visits problem}

\begin{definition}[$\kV$]\label{def:kV}
    Given a non-decreasing sequence of $n$ positive integers (deadlines) $d_i$, $i\in [n]$, the $\kV$ problem asks whether there exists a schedule of length $nk$, containing each $i \in [n]$ exactly $k$ times, with the constraint that every occurrence of $i$ is at most $d_i$ positions away from the previous one (or from the beginning of the schedule, if it is the first one).\footnote{By ``beginning of the schedule'' we refer to position $0$. We assume that the $nk$ visits of the nodes occur from position $1$ to position $nk$.}
\end{definition}

The $\kV$ problem can be modeled by an agent traversing a complete graph with loops where all edges have weight $1$ and each node $i$ has a deadline $d_i$. In each time unit, the agent can traverse an edge and renew the deadline of the node it visits, with the goal being to visit each node a total of $k$ times without its deadline expiring. With this in mind, we will refer to the numbers $i \in [n]$ as \emph{nodes} throughout the paper. 

\begin{observation}\label{obs:density}
    If a sequence of $n$ positive integers does not have a feasible schedule for the $\kV$ problem, then it does not have a feasible schedule for any $k'\textnormal{\textsc{-Visits}}$ problem with $k' > k$. Moreover, it has no feasible schedule for $\PWS$, which can (informally) be seen as the $\infty\textnormal{\textsc{-Visits}}$ problem.
\end{observation}

\begin{corollary}
    A $\kV$ instance with density not exceeding $5/6$ admits a feasible schedule, due to the density threshold conjecture proven for $\PWS$ in \cite{Kawamura_5/6_stoc}.
\end{corollary}

\begin{observation}
    The converse of Observation~\ref{obs:density} does not hold. Notably, sequences with density exceeding $1$ may have a feasible schedule for $\kV$ (for some positive integer $k$), even though they never have a feasible schedule for $\PWS$~\cite{Holte_Pinwheel}. For example, $\{1,2,\ldots,n\}$ admits a \oneV\ schedule.
\end{observation}

%\textcolor{red}{example}

The following definition will prove crucial for studying the $\kV$ problem.

\begin{definition}[Discretized Sequence]\label{def:disc}
    Given a non-decreasing sequence $D=\langle d_1,\ldots,d_n\rangle$ of positive integers, we define its \emph{discretized sequence} $A=\langle a_1,\ldots,a_n\rangle$ as 
    \[a_i = \begin{cases}
        d_i,\ i=n\\
        \min\{a_{i+1}-1, d_i\},\ i<n
    \end{cases}.\]
\end{definition}

Intuitively, the discretized sequence of a sequence of deadlines $d_i$ contains the latest possible positions in which the first visits of all $n$ nodes can occur. For example, the discretized sequence of $\langle 6,8,8,8,11,11,14\rangle$ is $\langle 5,6,7,8,10,11,14\rangle$. A discretized sequence can be computed in $\bigO(n)$ time.%, by scanning the input in decreasing order.%\todo{M: I suggest we re-write the part starting with "For example" as follows so that we can also refer to it as a sequence that has a schedule even though density is $>1$}

%\textcolor{blue}{
%\begin{example}\label{ex:disc-and-not-dense}
%    For example, the discretized sequence of $[6,8,8,11,11,14,23,23,23,23,23,23,23,23]$ is $[6,7,8,10,11,14,16,17,18,19,20,21,22,23]$. 
%\end{example}
%A discretized sequence can be computed in $\bigO(n)$ time, by scanning the input in decreasing order.
%}

Throughout the paper, we may refer to the input of $\kV$ as a \emph{multiset} instead of a sequence, when it is more convenient (or a \emph{set}, if it does not contain duplicates). For $\twoV$ specifically (which is the main focus of this paper), we will use the following definition. 

\begin{definition}[$\twoV$]\label{def:2V}
    Given a non-decreasing sequence of $n$ positive integers (deadlines) $d_i$, $i\in [n]$, the $\twoV$ problem asks whether there exists a schedule of length $2n$, containing a \emph{primary} and a \emph{secondary} visit for each $i \in [n]$. For every $i \in [n]$, its primary visit must be at most $d_i$ positions away from the beginning of the schedule and its secondary visit must be either before its primary visit or at most $d_i$ positions after its primary visit.
\end{definition}

Note that Def.~\ref{def:2V} is equivalent to Def.~\ref{def:kV} for $k=2$ and will significantly simplify many of our proofs by allowing a secondary visit to appear before the respective primary visit (in contrast to using \emph{first} and \emph{second} visits).

\section{The 1-Visit problem}\label{sec:one_visit}

\begin{lemma}\label{lemma:one-visit}
    If a $\oneV$ instance $\langle d_1,d_2,\ldots,d_n\rangle$ has a feasible schedule, then the schedule $1,2,\ldots,n$ is feasible for that instance.
\end{lemma}

\begin{proof}
Let us assume that a feasible schedule $S'$ exists and that the schedule $S = 1,2,\ldots,n$ is not feasible. Let $p$ be the first position where $S$ and $S'$ differ. Let $i$ be the node in position $p$ of $S$, $j$ the node in position $p$ of $S'$ and $p'>p$ be the position of $S'$ where node $i$ is placed. Then, $i$ and $j$ can be swapped in $S'$, as by definition $d_i  \leq d_j$. Iteratively, the same swapping argument holds for every position where $S$ and $S'$ differ and, as a result, $S'$ can be transformed to $S$ without violating a deadline at any step. Therefore $S$ is also a feasible schedule, contradiction.
\end{proof}

\begin{lemma}\label{obs:one-visit}
    A $\oneV$ instance has a feasible schedule if and only if the discretized sequence of the input contains only (strictly) positive elements.
\end{lemma}

The proof of this lemma follows immediately from Lemma~\ref{lemma:one-visit} and Def.~\ref{def:disc}, and is thus omitted. %\textcolor{red}{is this better?}

\begin{corollary}
    $\oneV$ can be solved in $\bigO(n)$ time.
\end{corollary}

\section{Tractable special cases of 2-Visits}\label{sec:two_visits_algo}

Trivial algorithms, such as greedily visiting the node with minimum deadline, fail for the $\twoV$ problem. In fact, as we will see later, there are instances of $\twoV$ where \emph{neither the first nor the second visits appear in order of non-decreasing deadline in any feasible schedule}.

\begin{definition}[Induced deadline]\label{def:induced_deadline}
    For a node $i$ whose \emph{primary} visit occurs in position $t_i$ of a schedule, we define its \emph{induced} deadline as $d_i'=d_i + t_i$. Intuitively, this number dictates the latest position in which the \emph{secondary} visit of node $i$ is allowed to occur.
\end{definition}

%For a node $i$ whose first visit occurs in position $t_i$ of the schedule, we define its \emph{induced} deadline as $d_i'=d_i + t_i$. Intuitively, this number dictates the latest position in which the second visit of node $i$ is allowed to occur.
A key intuition for any algorithm is to push primary visits as late as possible to maximize the induced deadlines. %To this effect, we push secondary visits as early as possible when allowed.
This intuition, although crucial, is insufficient for solving the $\twoV$ problem. It turns out that, for certain instances, the primary visits must be permuted in a counterintuitive manner in order to allow a feasible schedule. This is demonstrated in the following example.

\begin{example}\label{example:non-sorted-first}
    Let $\langle 4,5,6,7,8,8,10,10,11,15,22,23\rangle $ be a $\twoV$ instance. One can observe that it is impossible to construct a feasible schedule for this instance by choosing primary visits in non-decreasing order. This happens because nine primary visits must be done until time unit $11$, leaving space for only two secondary visits until that time unit. A reasonable choice is to visit the first two nodes twice each as early as possible; however, this would only make the induced deadline of the third node equal to $11$, which causes a problem because there is already a primary visit that must be placed at position $11$. A correct order of primary visits is actually $4,5,8,6,7,8,10,10,11,15,22,23$; it is not hard to verify that there exists a feasible schedule respecting this order.
\end{example}

At first glance, the aforementioned permutation of primary visits in Example~\ref{example:non-sorted-first} seems somewhat arbitrary. Surprisingly, every feasible schedule of that instance requires the nodes with deadlines $6,7,8,8$ to be visited in the order $8,6,7,8$. In order to systematically study permutations of primary visits, we will later introduce the $\PM$ problem, which serves as a key component of both our algorithms and hardness results.

\subsection{Properties of feasible 2-Visits schedules}

In the proofs of this subsection we will transform feasible schedules to make them satisfy certain useful properties, thus greatly limiting the amount of schedules that have to be considered by a $\twoV$ algorithm. %An issue with this approach is that many of the transformations we propose might place the second visit of a node before its first visit, which does not make sense by Def.~\ref{def:kV}. We solve this by defining \emph{primary} and \emph{secondary} visits and giving an alternative (equivalent) definition of $\twoV$.
%Since we are now using this alternative definition, we redefine \emph{induced} deadlines using primary/secondary visits instead of first/second visits.
By Lemma~\ref{obs:one-visit}, an input with a discretized sequence containing a non-positive number cannot have a feasible schedule for $\oneV$, and thus neither for $\twoV$ (by Obs.~\ref{obs:density}). From now on, we assume that the input always has a discretized sequence of positive numbers. We will also assume that the input contains no deadlines greater than $2n$. If it does, then the respective node would never expire and, thus, its two visits can be placed at the last two positions of the schedule. Clearly, a feasible schedule exists if and only if the input minus that node has a feasible schedule.

\begin{definition}[Gap]
    Let $A$ be the discretized sequence of a $\twoV$ instance. We call a position $i \in [2n]$ a \emph{gap} if $i \notin A$.
\end{definition}

With the above assumptions, any $\twoV$ instance has exactly $n$ gaps in $[2n]$. In Example~\ref{example:non-sorted-first}, the discretized sequence is $\langle 3,4,5,6,7,8,9,10,11,15,22,23\rangle $; hence, the $n=12$ gaps of this instance are $\{1,2,12,13,14,16,17,18,19,20,21,24\}$. 

A key idea that we will use in our approach is that placing all primary visits in the positions dictated by the discretized sequence is, in a certain sense, the best choice for obtaining a schedule. This crucial observation is formalized in Lemma~\ref{lemma:gaps} and will play a major role for the algorithms we propose in this section, as well as for the proof of strong NP-completeness of $\twoV$ in Section~\ref{sec:two_visits_hardness}.

Before stating Lemma~\ref{lemma:gaps}, we prove two auxiliary lemmas.

\begin{lemma}\label{lemma:disc_seq}
    Let $A=\langle a_1,\ldots,a_n\rangle $ be the discretized sequence of a $\twoV$ instance $D=\langle d_1,\ldots,d_n\rangle $. For any $i \in [n]$ it holds that at most $n-i$ primary visits can occur after position $a_i$.
\end{lemma}

\begin{proof}
    We use induction on $i$, with $i=n$ as base.

    \emph{Basis $(i=n)$.} We have $d_n = a_n$, hence no primary visit may occur after position~$a_n$.

    \emph{Inductive step $(i<n)$.} Suppose that at most $n-i-1$ primary visits occur after position~$a_{i+1}$. By Def.~\ref{def:disc}, we need to consider the following two cases.
    
    \emph{Case 1:} $a_i =a_{i+1}-1$. Since $a_i,\ a_{i+1}$ are consecutive positions, it follows immediately from the induction hypothesis that at most $n-i$ primary visits can occur after $a_i$.
    
    \emph{Case 2:} $a_i =d_i$. Primary visits of nodes $j \leq i$ cannot be placed after position $a_i$, since $d_j \leq d_i = a_i$. Hence, at most $n-i$ primary visits can occur after $a_i$.
\end{proof}

\begin{lemma}\label{lemma:auxiliary}
    If a feasible schedule places some secondary visit in a non-gap position $p$, then there exists some primary visit in a position $q<p$ belonging to a node $u$ with $d_u \geq p$.
\end{lemma}

\begin{proof}    
    Since $p$ is not a gap, it holds that $p \in A$. Suppose $p$ is the $i$-th element of $A$, i.e., $p = a_i$.
    By Lemma~\ref{lemma:disc_seq}, there are at least $i$ primary visits before position $p$ (not including $p$, by assumption). By the pigeonhole principle, one of these must have deadline at least $d_i$. By Def.~\ref{def:disc}, $p = a_i \leq d_i$.
\end{proof}

\begin{lemma}\label{lemma:gaps}
    If a $\twoV$ instance has a feasible schedule, then it has a feasible schedule in which all secondary visits are placed in gaps. (Equivalently: no primary visit is placed in a gap.)
\end{lemma}

\begin{proof}
    Suppose a $\twoV$ instance has a feasible schedule which places the secondary visit of some node $v$ in a non-gap position $p$. We will modify that schedule to force the secondary visit of $v$ to be placed in a gap.
    
    By Lemma~\ref{lemma:auxiliary}, there is a primary visit in a position $q<p$ belonging to a node $u$ with $d_u \geq p$. This allows us to swap the entries in positions $q,\ p$ without affecting the feasibility of the schedule; the primary visit of $u$ can be placed in position $p$, since $d_u \geq p$, and the secondary visit of $v$ is moved to an earlier position, preserving feasibility. Note that this argument holds even if $v=u$.\footnote{This is one of the reasons why we gave Def.~\ref{def:2V} for the $\twoV$ problem. This lemma would become significantly more complicated if Def.~\ref{def:kV} was used instead.}

    If by the above modification the secondary visit of $v$ was again placed in a non-gap position, we can apply Lemma~\ref{lemma:auxiliary} iteratively and do the same modification. Note that this modification moves the secondary visit of $v$ to a (strictly) earlier position, therefore it can only be applied a finite number of times. This implies that this procedure will terminate at some point, which can only happen if the secondary visit of $v$ is placed in a gap. Note that an appropriate gap must exist (by Lemma~\ref{lemma:auxiliary}).

    We can apply this procedure to all nodes whose secondary visits are not in gaps. Note that the procedure never displaces secondary visits that have been placed in gaps, therefore, by doing this, we will end up with a feasible schedule satisfying the desired property.
\end{proof}

We now present the main theorem of this subsection.

%\textcolor{red}{first sentence of proof needs to be reworded}

\begin{theorem}\label{theorem:second_visits}
    If a $\twoV$ instance has a feasible schedule, then it has a feasible schedule in which:
    
    \begin{enumerate}
        \item All secondary visits are placed in gaps. (Equivalently: no primary visit is placed in a gap.)
        \item Secondary visits appear in the schedule in order of non-decreasing induced deadlines.
    \end{enumerate}
\end{theorem}

\begin{proof}
    By Lemma~\ref{lemma:gaps}, we know that there is a feasible schedule satisfying the first property. Take such a schedule and reorder its secondary visits by non-decreasing induced deadline (primary visit positions are not affected by this modification). We can prove that the new schedule is also feasible through swapping arguments similar to the proof of Lemma~\ref{lemma:one-visit} (using induced deadlines instead of regular deadlines). 
\end{proof}

\subsection{Reducing 2-Visits to Position Matching}

Here, we will prove that \twoV\ reduces to solving a numerical matching variant for each \emph{cluster} of the discretized sequence.

\begin{definition}[Cluster]\label{def:cluster}
    We call a maximal subsequence of consecutive numbers of the discretized sequence of a $\twoV$ instance a \emph{cluster}. For two clusters $C_1,\ C_2$ of a discretized sequence, we say that $C_1$ \emph{precedes} $C_2$ if the elements of $C_1$ are smaller than the elements of $C_2$.
\end{definition}

%\begin{observation}
    %Every non-gap position of a schedule belongs to some specific cluster of $A$. All clusters are disjoint.
%\end{observation}

\begin{example}\label{example:clusters}
    Consider input $D=\langle 6,8,8,8,11,11,14\rangle$ with $n=7$. Then, $A=\langle 5,6,7,8,10,11,14\rangle$ is made up of three clusters: $C_1=\langle 5,6,7,8\rangle$, $C_2=\langle 10,11\rangle$ and $C_3=\langle 14\rangle $. Every $i\in [2n]$ that is not in any of these clusters is a gap.
\end{example}

%\textcolor{red}{define sequence with dots etc.}

\begin{observation}\label{obs:clusters}
    If $\langle a_p,\ldots,a_q\rangle$ is a cluster, then it is the discretized sequence of $\langle d_p,\ldots,d_q\rangle$. For example, $C_1=\langle 5,6,7,8\rangle$ is the discretized sequence of $\langle 6,8,8,8\rangle$ in Example~\ref{example:clusters}. Similarly, $C_2=\langle 10,11\rangle$ is the discretized sequence of $\langle 11,11\rangle$. The reason that this property holds is that $d_q = a_q$, by Definitions~\ref{def:disc},~\ref{def:cluster}.
\end{observation}

\begin{lemma}\label{lemma:clusters}
    Let $\langle a_i,\ldots,a_j\rangle$ be a cluster of the discretized sequence $A$ of a $\twoV$ input. For any feasible schedule that satisfies the properties of Theorem~\ref{theorem:second_visits}, it holds that the primary visits of nodes $i,\ldots,j$ are placed in some permutation of the positions in $\langle a_i,\ldots,a_j\rangle$.
\end{lemma}
    
\begin{proof}
    If $A$ consists of only one cluster, then the lemma follows immediately from Theorem~\ref{theorem:second_visits}. Assume $A$ consists of more than one cluster.

    Let $F$ be a feasible schedule that satisfies the properties of Theorem~\ref{theorem:second_visits} and let $C_1 = \langle a_1,\ldots,a_m\rangle$ be the first cluster of $A$. Since $C_1$ is a cluster, it holds that $a_m < a_{m+1} - 1$. By Def.~\ref{def:disc}, this implies that $a_m = d_m$ and, thus, $d_m < a_{m+1} - 1$. Since $D=\langle d_1,\ldots,d_n\rangle$ is a non-decreasing sequence, we obtain that: 

    \begin{equation}\label{eq1}
        d_i < a_{m+1} - 1, \ \forall i \in [m]
    \end{equation}

    Recall that in a schedule satisfying the properties of Theorem~\ref{theorem:second_visits}, no primary visit is in a gap. Thus, all primary visits must be placed at some $a_i \in A$. By~(\ref{eq1}), we know that the first $m$ primary visits cannot be placed in $a_{m+1}$ or later. Thus, all of them must be placed in some permutation of the positions $a_1,\ldots,a_m$.

    We can prove the lemma through induction, using the above proof for $C_1$ as the induction basis. In each step, we study a cluster $C_j=\langle a_p,\ldots, a_q\rangle $. By induction hypothesis, we know that all $a_i,\ i \in [p-1]$, are occupied by primary visits of nodes $i\in [p-1]$. Exactly the same arguments apply to show that primary visits of nodes $i \in [p, q]$ cannot be placed in any position greater than $a_{q}$, thus forcing them to be placed in some permutation of the positions in $\langle a_p,\ldots,a_q\rangle $, which completes the proof. 
\end{proof}

By Theorem~\ref{theorem:second_visits} and Lemma~\ref{lemma:clusters}, we know that in order to find a feasible schedule, it suffices to:
\begin{itemize}
    \item Place the primary visits of nodes corresponding\footnote{We say that nodes $i,\ldots,j$ \emph{correspond} to cluster $C=\langle a_i,\ldots,a_j\rangle $.} to each cluster in some permutation of the positions of that cluster.
    \item  Place secondary visits in gaps in order of non-decreasing induced deadline.
\end{itemize}

Based on this information, if all $n$ elements of the discretized sequence form a single cluster, then a brute-force algorithm would run in $\bigO(n! \cdot n\log n)$ time, by checking all permutations of primary visits and then sorting the nodes by non-decreasing induced deadline for each permutation.
In order to investigate whether this complexity can be improved, we introduce the $\PM$ problem, which captures the question of how primary visits should be permuted.

\begin{definition}[$\PM$]\label{def:SM}
    Given a non-decreasing sequence $D=\langle d_1,\ldots,d_n\rangle$ of positive integers, its discretized sequence $A=\langle a_1,\ldots,a_n\rangle $ and a set $T=\{t_1,\ldots,t_n\}$ of $n$ distinct positive integers (targets), the $\PM$ problem asks whether there is a subset $M$ of $D\times A\times T$ s.t. every $d_i \in D$, $a_i \in A$, $t_i \in T$ occurs exactly once in $M$ and for every triplet $(d,a,t)\in M$ it holds that $d \geq a$ and $d+a\geq t$.
\end{definition}

We would like to emphasize the restriction that elements of $D$ cannot be matched with larger elements of $A$. Throughout the paper, we will say that a tuple $(d,a)$ \emph{satisfies} a target $t$ if $d+a \geq t$. %We present an example instances of $\PM$, which is useful for understanding this relatively complicated problem.

\begin{example}
    Let $D=\langle 6,7,8,8,15,15\rangle $ and $T=\{ 12,13,14,15,20,28\} $. We compute the discretized sequence $A=\langle 5,6,7,8,14,15\rangle$ of $D$. A solution to this $\PM$ instance can be obtained by using the following $D$-$A$ pairs: $(6,6),(8,5),(7,7),(8,8),(15,14),(15,15)$. The sums of these pairs are $12,13,14,16,29,30$, satisfying all targets. Thus, a solution to this instance consists of the triplets $(6,6,12),(8,5,13),(7,7,14),(8,8,15),(15,14,20),(15,15,28)$.
\end{example}

%\textcolor{red}{cut one example?}

%\begin{example}
    %Let $D=\langle 6,7,8,8\rangle$ and $T=\{9,14,15,16\}$. We compute the discretized sequence $A=\langle 5,6,7,8\rangle$ of $D$. Observe that this $\PM$ instance has no solution. Targets $15$ and $16$ can only be achieved through $8+7$ and $8+8$, rendering $14$ impossible to satisfy.
%\end{example}

\begin{remark}
    Although $A$ can be derived from $D$ in Def.~\ref{def:SM}, we assume it is given as input for the sake of readability. In general, throughout the paper, whenever a variant of $\NTDMshort$ is discussed we consider all three sets as input, even if some of them can be implicitly derived. This makes our reductions in the next section more natural and easier to read.
\end{remark}

\begin{lemma}\label{lemma:secondary_sorted_by_cluster}
    Let $C_1$ and $C_2$ be two clusters with $C_1$ preceding $C_2$. If a schedule satisfies the properties of Theorem~\ref{theorem:second_visits}, then for nodes $u,v$ such that $a_u \in C_1$ and $a_v \in C_2$ the following inequality holds regarding their induced deadlines in that schedule: $d_u' < d_v'$.
\end{lemma}

\begin{proof}
    By Lemma~\ref{lemma:clusters}, we know that, if $\langle a_i,\ldots,a_j\rangle $ is a cluster, a schedule satisfying the properties of Theorem~\ref{theorem:second_visits} places primary visits of nodes $i,\ldots,j$ in some permutation of the positions in $\langle a_i,\ldots,a_j\rangle $. Thus, if a schedule satisfies the properties of Theorem~\ref{theorem:second_visits}, it holds that for all $a_p \in C_1$ and $a_q \in C_2$, positions $a_p$,\ $a_q$ contain primary visits of some nodes $u,v$ respectively, with $u<v$ and thus $d_u \leq d_v$. Combined with the fact that $a_p < a_q$, this implies that  $d_u' < d_v'$.
\end{proof}

Theorem~\ref{theorem:reduction} is the backbone of the algorithms we propose for $\twoV$ in the next subsections.

\begin{theorem}\label{theorem:reduction}
    $\twoV$ reduces in linear time to solving a $\PM$ instance for each cluster.
\end{theorem}
    
\begin{proof}
    By Theorem~\ref{theorem:second_visits}, we know that it suffices to place all secondary visits in gaps by non-decreasing induced deadline in order to find a feasible $\twoV$ schedule. By Lemma~\ref{lemma:secondary_sorted_by_cluster}, we know that, for a schedule satisfying the properties of Theorem~\ref{theorem:second_visits} and nodes $u,v$ such that the cluster containing $a_u$ precedes the cluster containing $a_v$: $d_u' < d_v'$. From all the above, we infer that, in such a schedule, the secondary visits of nodes corresponding to the first cluster $C_1$ of $A$ are placed in the first $|C_1|$ gaps of the schedule (by sorted induced deadline). Similarly, the secondary visits of nodes corresponding to the second cluster $C_2$ are placed in the next $|C_2|$ gaps, and so on.

    The above, in conjunction with Lemma~\ref{lemma:clusters}, implies that each cluster can be solved independently; the positions in which primary and secondary visits of nodes corresponding to a cluster must be placed do not coincide with the respective positions of other clusters. We will say that there is a \emph{partially feasible} schedule for a cluster $C$ if all of its corresponding nodes can be placed in these designated positions in an order that satisfies all their deadlines (for primary visits) and all their induced deadlines (for secondary visits).

    For a cluster $C=\langle a_p,\ldots,a_q\rangle $, let $t_1,t_2,\ldots,t_{q-p+1}$ be the gaps %\footnote{These gaps are easily found in linear time, if we process clusters by increasing elements. They are always the first $|C|$ unused gaps, where $|C|$ is the size of the current cluster.}
    in which the secondary visits of its corresponding nodes must be placed, according to the above. The primary visits of its corresponding nodes must be placed in positions $a_p,\ldots,a_q$, according to Lemma~\ref{lemma:clusters}. The position $a_i$ of a primary visit of a node $v$ is feasible if $d_v \geq a_i$. The secondary visits of nodes $p,\ldots,q$ can be placed feasibly if and only if there is a node with induced deadline at least $t_1$ (to be inserted into the gap $t_1$), another one with induced deadline at least $t_2$, and so on until $t_{q-p+1}$. It turns out that this is a $\PM$ instance with input sequence $D=\langle d_p,\ldots,d_q\rangle $, discretized sequence $C=\langle a_p,\ldots,a_q\rangle $\footnote{$C$ is the discretized sequence of $D$ by Observation~\ref{obs:clusters}.} and targets $T=\{t_1,t_2,\ldots,t_{q-p+1}\}$; deadlines have to be matched with smaller positions of the cluster and yield sufficiently large induced deadlines, for secondary visits to be placed in gaps. %Observe that the $d \geq a$ restriction of $\PM$ for pairs $(d,a)\in D\times A$ captures the $d_v \geq a_i$ restriction explained above for primary visits.
    This $\PM$ instance has a feasible solution if and only if there is a partially feasible schedule for cluster $C$.

    A feasible $\twoV$ schedule (satisfying the properties of Theorem~\ref{theorem:second_visits}) exists if and only if every cluster has a partially feasible schedule. Thus, by Theorem~\ref{theorem:second_visits}, solving $\twoV$ is reduced to solving a $\PM$ instance for each cluster.
\end{proof}

We now show how to solve the $\twoV$ instance from Example~\ref{example:non-sorted-first} using the reduction of Theorem~\ref{theorem:reduction}. Figure \ref{fig:matching} shows the \PM\ instances and Figure \ref{fig:schedule} shows the resulting schedule.

\begin{figure}[ht]
    \centering
    \includegraphics[scale=1.1]{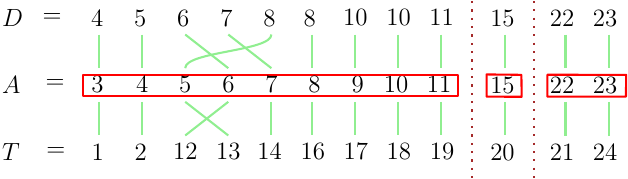}
    \caption{A $\twoV$ instance $D$, with its discretized sequence $A$ consisting of three clusters (included in red boxes). $T$ is the respective set of targets (gaps). Green lines among $D,A,T$ show triplets that constitute solutions to each of the three \PM\ instances obtained through Theorem~\ref{theorem:reduction}.}
    \label{fig:matching}
\end{figure}

\begin{figure}[ht]
    \centering
    \includegraphics[scale=1.2]{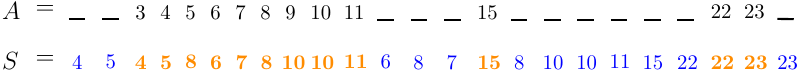}
    \caption{The sequence $A$ of Fig.~\ref{fig:matching} with its gaps and the schedule $S$ found by solving the \PM\ instances. Primary visits are denoted with orange (bold) and secondary visits with blue. For simplicity, we use the nodes' deadlines to show the schedule (instead of their indices).}
    \label{fig:schedule}
\end{figure}

\subsection{2-Visits with distinct deadlines}

%Sketch: For distinct deadlines, $A=D$, rendering sum matching trivial (it has only one solution that can be feasible: to match every element with its equal). Thus, it can be solved in linear time. The induced deadlines are (automatically) sorted just like the input deadlines are sorted, therefore no sorting is required. Runs in linear time. Obs \ref{obs:clusters} may be useful here.

Let $D=\langle d_1,\ldots,d_n\rangle $ be a sequence containing distinct deadlines. For its discretized sequence $A=\langle a_1,\ldots,a_n\rangle $ it holds that $A=D$, by Definition~\ref{def:disc}. By Theorem~\ref{theorem:reduction}, the $\twoV$ instance $D$ reduces in linear time to solving a $\PM$ instance of the form $I=(D'=\langle d_p,\ldots,d_q\rangle ,A'=\langle a_p,\ldots,a_q\rangle ,T)$ for each cluster $A'$. Observe that $D'=A'$, which renders $I$ trivial through the restriction that $d\in D'$ cannot be matched with $a\in A'$ s.t. $d<a$: $d_p$ can only be matched with $a_p$ and, thus, $d_{p+1}$ can only be matched with $a_{p+1}$ and so on. The $\PM$ instance $I$ has a solution if and only if this unique feasible matching between $D'$ and $A'$ satisfies every target in $T$, leading to the following lemma.

\begin{lemma}
    A $\PM$ instance $(D,A,T)$ with $D$ containing only distinct numbers can be solved in $\bigO(n)$ time.
\end{lemma}

To construct the respective feasible schedule for the $\twoV$ instance $D$, we have to place the secondary visits in gaps by increasing induced deadline, according to Theorem~\ref{theorem:second_visits}. However, sorting them by induced deadline in this case is unnecessary, since the aforementioned matching always places the primary visits by increasing deadline, meaning that the nodes are already in the desired order. Thus, we obtain the following theorem.

\begin{theorem}\label{thrm:distinct_linear}
    A $\twoV$ instance with distinct deadlines can be solved in $\bigO(n)$ time. A respective schedule (if it exists) can be constructed in $\bigO(n)$ time as well.
\end{theorem}

\subsection{2-Visits with constant maximum cluster size}\label{subsec:cluster}

Let $D=\langle d_1,\ldots,d_n \rangle$ be a \twoV\ instance with discretized sequence $A$ consisting of $m$ clusters with sizes $c_1,\ldots,c_m$, where $\sum_{i=1}^{m} c_i = n$. Let the \emph{maximum cluster size} of $A$ be $c=\max_{i\in [m]}{c_i}$. We will show that \twoV\ is FPT parameterized by $c$.

First, let us observe that \PM\ can be solved by brute force in $\bigO(n! \cdot n\log n)$ time, since it suffices to compute all feasible matchings between $D$ and $A$, and then, for each feasible matching, sort the respective pairs in order of non-decreasing sums and match them with the targets in $T$ in increasing order. It is clear that if no solution exists with this choice, then the \PM\ instance admits no solution, since a target that is not satisfied by some pair $(d,a)$ cannot be satisfied by any pair $(d',a')$ with $d'+a'<d+a$, by definition of the problem. Thus, we obtain the following.

\begin{observation}\label{obs:PM_brute_force}
    \PM\ admits a brute-force algorithm running in time $\bigO(n! \cdot n\log n)$. 
\end{observation}

With this, we can obtain the following FPT algorithm for \twoV.

\begin{theorem}\label{thrm:FPT_cluster_brute_force}
    \twoV\ admits an FPT algorithm parameterized by the maximum cluster size $c$ of the input's discretized sequence, running in time $\bigO\left(nc!\log c\right)$.
\end{theorem}

\begin{proof}
    By Theorem~\ref{theorem:reduction}, we have that $\twoV$ reduces in linear time to solving $m$ \PM\ instances with sizes $c_1,\ldots,c_m$. Using Observation~\ref{obs:PM_brute_force} for each of these \PM\ instances, we obtain an algorithm for \twoV\ running in time $\bigO\left(\sum_{i=1}^{m} c_i! \cdot c_i\log c_i\right) = \bigO\left(c!\cdot\log c \cdot\sum_{i=1}^{m} c_i\right) = \bigO\left(nc!\log c\right)$.
\end{proof}

However, we can further improve on the running time of the algorithm of Theorem~\ref{thrm:FPT_cluster_brute_force} by employing dynamic programming (DP) for solving the reduced \PM\ instances, as follows.

Let $I=(D,A,T)$ be a \PM\ instance, where $D=\{ d_1,\ldots,d_n\}$.
Create a DP table with states $(i,A_{used},T_{used})$, where $i\in \{0,\ldots,n\}$ represents the number of elements of~$D$ that have been processed, and $A_{used},\ T_{used}$ are the subsets of $A$ and $T$ that have been used for the matching produced up to that point. For each state in the DP table, we store $i$ triplets $(d,a,t)\in \{d_1,\ldots,d_i\}\times A_{used} \times T_{used}$ that induce a partial solution for $I$ (up to the $i$-th element of~$D$).
The initialization of the DP table is $(0,\emptyset,\emptyset)$. We next describe the transition between states. For each state $S=(i,A_{used},T_{used})$ and each pair $(a,t)\in (A\setminus A_{used})\times (T\setminus T_{used})$, if the triplet $(d_{i+1},a,t)$ is a valid \PM\ triplet (i.e., it satisfies $d_{i+1}\geq a$ and $d_{i+1}+a\geq t$), then we transition to the state $S'=(i+1,A_{used}\cup \{a\},T_{used}\cup \{t\})$. The solution stored for the new state $S'$ is obtained by appending the triplet $(d_{i+1},a,t)$ to the one stored for $S$. This concludes the description of the algorithm. The algorithm's correctness follows immediately from the definition of the problem, and is thus omitted.

We now analyze the running time of the described DP algorithm for \PM. For each $i\in [n]$, the number of states of the DP table is ${\binom{n}{i}}^2$, since the respective sets $A_{used}$ and $T_{used}$ consist of $i$ elements each. For each state $(i,A_{used},T_{used})$, the number of pairs in $A\times T$ that have to be checked is bounded by $(n-i)^2$, since $|A\setminus A_{used}|=|T\setminus T_{used}|=n-i$. Thus, we obtain a running time bound of
\[
    \bigO\left(\sum_{i=1}^n {(n-i)^2\binom{n}{i}}^2 \right) =
    \bigO\left(n^2 \binom{2n}{n}\right) =
    \bigO(4^n n^{3/2}),
\]
where the last step is obtained through $\binom{2n}{n}=\Theta(4^n/\sqrt{n})$ by Stirling's approximation.

From the above, we obtain the following lemma.

\begin{lemma}\label{lem:PM_DP}
    \PM\ admits a dynamic programming algorithm running in time $\bigO(4^n n^{3/2})$.
\end{lemma}

With Lemma~\ref{lem:PM_DP}, we can now improve upon the running time of Theorem~\ref{thrm:FPT_cluster_brute_force}.

\begin{theorem}\label{thrm:FPT_cluster_DP}
    \twoV\ admits an FPT algorithm parameterized by the maximum cluster size $c$ of the input's discretized sequence, running in time $\bigO\left(n4^c\sqrt{c}\right)$.
\end{theorem}

\begin{proof}
    By Theorem~\ref{theorem:reduction}, we have that $\twoV$ reduces in linear time to solving $m$ \PM\ instances with sizes $c_1,\ldots,c_m$. Using Lemma~\ref{lem:PM_DP} for each of these \PM\ instances, we obtain an algorithm for \twoV\ running in time $\bigO\left(\sum_{i=1}^m 4^{c_i} {c_i}^{3/2}\right) = \bigO\left(4^c\sqrt{c}\sum_{i=1}^m c_i \right) = \bigO\left(n4^c\sqrt{c}\right)$.
\end{proof}

\subsection{2-Visits with up to two distinct numbers per cluster}\label{subsec:two_distinct}

We will examine \PM\ instances with $D$ consisting of copies of up to two distinct numbers, and then transfer an algorithm to \twoV\ with up to two distinct deadlines per cluster through Theorem~\ref{theorem:reduction}.

\begin{observation}\label{obs:trivial_SM}
    A $\PM$ instance $(D,A,T)$ with $D$ only containing copies of a single number is trivial. It suffices to add that number to each $a \in A$ and check whether all $t \in T$ are satisfied.
\end{observation}

For the case of two distinct numbers, we will utilize the following auxiliary lemma, which handles the case in which some target is small enough to be satisfied by any pair $(d,a)$.

\begin{lemma}\label{lem:SM_small_target}
    If for a $\PM$ instance $I=(D,A,T)$ there exists some $t \in T$ s.t. $t \leq \min(D) + \min(A)$, then it suffices to match $t$ with $\min(D)$ and $\min(A)$; if there is no solution with this choice, then there is no solution for $I$.
\end{lemma}
%\textcolor{red}{new proof suggested by reviewer}

\begin{proof}
    Suppose $I$ has some solution $M$ that does not match $t$ with $\min(D)$ and $\min(A)$. We will prove that this implies the existence of a solution $M'$ for $I$ in which $t$ is matched with $\min(D)$ and $\min(A)$.

    First, consider the case in which $t$, $\min(D)$ and $\min(A)$ are all contained in different triplets in~$M$. Let $(d',a',t),(\min(D),a'',t'),(d'',\min(A),t'')$ be these triplets. We construct $M'$ by replacing these three triplets with the following three, while keeping the rest of $M$ intact: $(\min(D),\min(A),t),$ $(d',\max\{a',a''\},t'),(d'',\min\{a',a''\},t'')$. It remains to prove that these three triplets are feasible \PM\ triplets, i.e., that their second element is no larger than the first and the third element is no larger than the sum of the first two.

    \begin{itemize}
        \item By Def.~\ref{def:disc}, it holds that $\min(D)\geq \min(A)$. By assumption, $t \leq \min(D) + \min(A)$, hence the first triplet is feasible.
        
        \item By assumption, $(d',a',t)$ is a feasible triplet, implying $d'\geq a'$, and $(\min(D),a'',t')$ is a feasible triplet, implying $d' \geq \min(D)\geq a''$. Combining these, we get $d' \geq \max\{a',a''\}$. It also holds that $d'+\max\{a',a''\} \geq \min(D) + a'' \geq t'$ (with the latter inequality holding because $(\min(D),a'',t')$ is a feasible triplet), thus yielding that $(d',\max\{a',a''\},t')$ is feasible.
        
        \item By assumption, $(\min(D),a'',t')$ is a feasible triplet, implying $d'' \geq \min(D)\geq a'' \geq \min\{a',a''\}$. Additionally, $(d'',\min(A),t'')$ is a feasible triplet, implying that $d''+\min\{a',a''\} \geq d''+\min(A) \geq t''$, thus yielding that $(d'',\min\{a',a''\},t'')$ is feasible.
    \end{itemize}

    The proof for the case in which two out of the three elements $t$, $\min(D)$, $\min(A)$ are contained in the same triplet in $M$ is analogous (and simpler) and is thus omitted.
\end{proof}

% \begin{corollary}\label{cor:one_deadline}
%     A $\twoV$ instance that only contains copies of a single number can be solved in $\bigO (n)$ time. \textcolor{red}{to be deleted}
% \end{corollary}

We now discuss the case where $D$ consists of only two distinct deadlines in a \PM\ instance $(D,A,T)$. That is, $D=\langle d_1,\ldots,d_n\rangle $ contains $m$ copies of $x$ and $n-m$ copies of $y>x$.

%Let $A=\langle a_1,\ldots,a_n\rangle $ be the discretized sequence of $D$. For $i \neq m$ it holds that $a_{i+1} = a_i +1$, since $d_i = d_{i+1}$. We will show an algorithm for $\twoV$ with two distinct deadlines by considering the cases $a_{m+1} \neq a_m +1$ and $a_{m+1} = a_m +1$ separately.

% \begin{lemma}\label{lem:2distinct_2clusters}
%     A $\twoV$ instance $D=\langle d_1,\ldots,d_n\rangle $ with $d_1=d_2=\cdots=d_m = x$, $d_{m+1}=\cdots=d_n = y$, $x<y$, can be solved in $\bigO(n)$ time if $a_{m+1} \neq a_m +1$.
% \end{lemma}

% \begin{proof}
%     Since all $a_i \in A$ except $a_m, a_{m+1}$ are consecutive numbers, $A$ consists of exactly two clusters: $C_1 =\langle a_1,\ldots,a_m\rangle $ and $C_2 = \langle a_{m+1},\ldots,a_n\rangle $. By Theorem~\ref{theorem:reduction}, it remains to solve two \PM\ instances, one for the first $m$ elements of $D$ and $A$ and another for the remaining $n-m$ elements of $D$ and $A$. The first $\PM$ instance receives as input a set $D'$ containing only copies of $x$, rendering it trivial by Observation~\ref{obs:trivial_SM}. The same holds for the second $\PM$ instance with $y$.
% \end{proof}

\begin{lemma}\label{lem:2distinct_1cluster}
    A \PM\ instance $I=(D=\langle d_1,\ldots,d_n\rangle,A,T)$ with  $d_1=d_2=\cdots=d_m = x$, $d_{m+1}=\cdots=d_n = y$, $x<y$, can be solved in $\bigO(n)$ time.

    %A $\twoV$ instance $D=\langle d_1,\ldots,d_n\rangle $ with $d_1=d_2=\cdots=d_m = x$, $d_{m+1}=\cdots=d_n = y$, $x<y$, can be solved in $\bigO(n)$ time if $a_{m+1} = a_m +1$.
\end{lemma}

\begin{proof}
    %All numbers in $A$ are consecutive, therefore $A$ is itself a cluster. By Theorem~\ref{theorem:reduction}, the $\twoV$ instance $D$ is equivalent to the $\PM$ instance $I=(D,A,T)$, with $T$ consisting of the first $n$ gaps, i.e., $T= T_1 \cup T_2$, where $T_1= \{1,\ldots,a_1 -1\}$, $T_2= \{a_1 + n, \ldots, 2n\}$.\footnote{Depending on the value of $a_1$, $T_1$ or $T_2$ might be empty sets. Observe that $a_n = a_1 + n - 1$.} Targets in $T_1$ are small enough to be satisfied by every pair $(d,a)\in D\times A$, so we can match them with the $a_1-1$ smallest elements of $D$ and the $a_1-1$ smallest elements of $A$ and remove all of these matched elements from the input, by Lemma~\ref{lem:SM_small_target}. Let $I'=(D',A',T_2)$ be the (equivalent) $\PM$ instance obtained after this step.

    %If all copies of $x$ in $D$ were used in the last step, $I'$ is trivial by Observation~\ref{obs:trivial_SM}. Suppose that at least one copy of $x$ remains in $D'$. 

    Let $a=\min(A)$, $t=\min(T)$ and check whether $a + x$ satisfies $t$. If it does, then $t$ is satisfied by any pair in $D \times A$ and  can thus be matched with $a$ and a copy of $x$ to obtain an equivalent $\PM$ instance, by Lemma~\ref{lem:SM_small_target}. If it does not, then $a$ cannot satisfy any target in $T$ when matched with a copy of $x$ and, therefore, must be matched with a copy of $y$. If $a+y < t$, then there is no feasible match for $a$, hence $I$ does not have a solution. If $a+y \geq t$, then we match~$a$ and a copy of $y$ with the largest target that they satisfy and remove them from the input. This choice is optimal, since any pair satisfies all targets smaller than one that it satisfies; thus, any solution that can be obtained by matching~$a$ and a copy of $y$ with some target that they satisfy, can also be obtained by matching them with the largest target that they satisfy.

    We can apply the above iteratively for each $a \in A$ (in increasing order). Note that, in subsequent repetitions, it might be the case that $a$ cannot be matched with $x$ due to $a > x$. If this occurs, and there is at least one unused copy of $x$ left, then $I$ clearly has no solution, since there are no feasible matches for $x$. Additionally, note that $a > y$ cannot occur, by definition of the problem. Excluding the aforementioned trivial case, by repeating the described procedure, in each step we either match $a$ with some $d \in D$ and some $t\in T$ or we conclude that $I$ has no solution. Since $a+y$ strictly increases in each iteration, finding the largest target satisfied by $a+y$ takes $\bigO(n)$ time for all iterations combined, if we store the index of the last such target found previously (in the initial form of the input set $T$). The process terminates when all copies of $x$ or all copies of $y$ in $D$ are used up, since the remaining instance is then trivial by Observation~\ref{obs:trivial_SM}. This solves $I$ in $\bigO(n)$ time.
\end{proof}

%We obtain the following theorem from Lemmas~\ref{lem:2distinct_2clusters},~\ref{lem:2distinct_1cluster} and Corollary~\ref{cor:one_deadline}.

Since \twoV\ reduces to solving a \PM\ instance for each cluster by Theorem~\ref{theorem:reduction}, we obtain the following through Observation~\ref{obs:trivial_SM} and Lemma~\ref{lem:2distinct_1cluster}.

\begin{theorem}\label{thrm:two_distinct_deadlines}
    A $\twoV$ instance that only has up to two distinct deadlines corresponding to each cluster of the input's discretized sequence can be solved in $\bigO (n)$ time.
\end{theorem}

\section{The computational complexity of 2-Visits}\label{sec:two_visits_hardness}

%We define a variation of $\PM$ without the restriction that elements in $D$ cannot be paired with larger elements of $A$, to be used as an intermediate step in our reductions.

%\begin{definition}[$\USM$]\label{def:USM}
    %Given a non-decreasing sequence $D$ of $n$ positive integers, its discretized sequence $A$ and a set $T=\{t_1,\ldots,t_n\}$ of $n$ distinct positive integers (targets), the $\USM$ problem asks whether there is a subset $M$ of $D\times A\times T$ s.t. every $d_i \in D$, $i \in A$, $t_i \in T$ occurs exactly once and for every triplet $(d,a,t)\in M$ it holds that $d+a\geq t$.
%\end{definition}

We define the following numerical matching problem, to be used as an intermediate step in our reduction. This problem may be of independent interest as a variation of $\NTDMshort$.

\begin{definition}[$\INTDMshort$]\label{def:numerical_ineq}
    Given a multiset of positive integers $A=\{a_1,\ldots,a_n\}$, the set $B=\{1,\ldots,n\}$ and a multiset of positive integers (targets) $T=\{t_1,\ldots,t_n\}$, the \INTDM\ $(\INTDMshort)$ problem asks whether there is a subset $M$ of $A \times B \times T$ s.t. every $a_i \in A$, $b \in B$, $t_i \in T$ occurs exactly once in $M$ and for every triplet $(a,b,t)\in M$ it holds that $a+b \geq t$.
\end{definition}

In this section we present a three-step reduction, from \RNTDMshort\ \cite{Flow_shop_Yu} to \INTDMshort\, to \PM\ and, finally, to \twoV. We consider this our main technical contribution.

\begin{remark}
    Membership in NP is trivial for most problems considered in this paper and is thus always omitted. The exception is $\PWS$ (and its generalizations in Section~\ref{sec:hardness_extra}), for which membership in NP, NP-hardness and PSPACE-completeness are all open questions ever since its introduction in 1989~\cite{Holte_Pinwheel}.
\end{remark}

\subsection{Reducing RN3DM to IN3DM}

We will use the following restricted version of $\NTDMshort$ for our reduction, defined by Yu, Hoogeveen and Lenstra~\cite{Flow_shop_Yu}. Its constraints are vital for the reduction to $\PM$ in the next subsection.

\begin{definition}[$\RNTDMshort$]\label{def:rn3dm}
    Given a multiset of positive integers $A=\{a_1,\ldots,a_n\}$, sets $B=C=\{1,\ldots,n\}$ and an integer $\sigma$ such that
    $$\sum_{i=1}^n (a_i+2i) = n\sigma,$$
    the $\RNTDM$ $(\RNTDMshort)$ problem asks whether there is a subset $M$ of $A \times B \times C$ s.t. every $a_i \in A$, $b \in B$, $c \in C$ occurs exactly once in $M$ and for every triplet $(a,b,c)\in M$ it holds that $a+b+c=\sigma$.
\end{definition}

\begin{theorem}[Yu et al. 2004~\cite{Flow_shop_Yu}]\label{thrm:num_hardness}
    $\RNTDMshort$ is strongly NP-complete.
\end{theorem}

\begin{theorem}\label{thrm:reduction_1}
    $\RNTDMshort$ reduces to $\INTDMshort$ in polynomial time.
\end{theorem}

\begin{proof}
    Let $I=(A,B,C,\sigma)$ be a $\RNTDMshort$ instance. Define the set $T=\{\sigma-i \mid i\in C\}$. By the $\RNTDMshort$ constraints, it holds that

    \begin{equation}
    \label{eq:ineq_eq}
            \sum_{i=1}^n (a_i+i) = n\sigma - \sum_{i=1}^n i = \sum_{i=1}^n t_i.
    \end{equation}

    Define the $\INTDMshort$ instance $I'=(A,B,T)$. $I'$ is a yes-instance of $\INTDMshort$ if and only if there is a subset $M$ of $A\times B \times T$ such that every $a_i \in A$, $b_i \in B$, $t_i \in T$ occurs exactly once in $M$ and for every triplet $(a,b,t)\in M$ it holds that $a+b\geq t$. By~(\ref{eq:ineq_eq}), ``$a+b\geq t$, $\forall(a,b,t)\in M$'' is equivalent to ``$a+b= t$, $\forall(a,b,t)\in M$''. Thus, equivalently, there is a subset $M'$ of $A\times B \times C$ such that every $a_i \in A$, $b_i \in B$, $c_i \in C$ occurs exactly once in $M'$ and for every triplet $(a,b,c)\in M'$ it holds that $a+b= \sigma-c$. Thus, $I'$ is a yes-instance of $\INTDMshort$ if and only if $I$ is a yes-instance of $\RNTDMshort$.
\end{proof}

\begin{lemma}\label{lem:range}
    $\INTDMshort$ is strongly NP-complete, even when:
    \begin{itemize}
        \item $\max(A)-\min(A) \leq 2n-2$,
        \item $t > \min(A)$, $\forall t\in T$, and
        \item $T$ is a simple set (instead of a multiset).
    \end{itemize}
\end{lemma}

\begin{proof}
    The reduction of Theorem~\ref{thrm:reduction_1} already forces $T$ to be a simple set. Additionally, it reduces to an \INTDMshort\ instance with $$\min(T)=\sigma - n = \frac{\sum_{i=1}^n a_i}{n}+1 > \min(A),$$
    thus also satisfying $t > \min(A)$, $\forall t\in T$.
    It remains to prove that the reduced instance in Theorem~\ref{thrm:reduction_1} can additionally be restricted to satisfy
    $\max(A)-\min(A) \leq 2n-2$.

    Let $I=(A,B,C,\sigma)$ be an $\RNTDMshort$ instance, where $B=C=\{1,\ldots,n\}$. In any $\RNTDMshort$ solution, the triplet containing $\max(A)$ cannot have a sum smaller than $\max(A)+2$ and the triplet containing $\min(A)$ cannot have a sum larger than $\min(A)+2n$. For a valid $\RNTDMshort$ solution, all triplets $(a,b,c)\in M$ must have sum equal to $\sigma$ and, thus, any $\RNTDMshort$ instance with $\max(A)-\min(A) > 2n-2$ is a trivial no-instance. Hence, we can assume that $\max(A)-\min(A) \leq 2n-2$ for the reduction of Theorem~\ref{thrm:reduction_1}. Note that the reduction does not modify $A$, which implies that the reduced $\INTDMshort$ instance will also satisfy the same property. In conjunction with Theorem~\ref{thrm:num_hardness}, this proves the lemma.
\end{proof}

\subsection{Reducing IN3DM to Position Matching}

For the following reduction we need to pay attention to the constraints of the $\PM$ problem (Def.~\ref{def:SM}). Specifically, for a $\PM$ instance $(D,A,T)$:

\begin{itemize}
    \item $A$ is the discretized sequence of $D$.
    \item $T$ only contains distinct positive integers.
    \item Elements in $D$ cannot be matched with larger elements of $A$.
\end{itemize}

Lemma~\ref{lem:range} will be useful for ensuring that these constraints are satisfied. The main challenge in the reduction of this subsection is to force the second set to be the discretized sequence of the first one, while simultaneously enforcing equivalence between the two instances despite the fact that $\INTDMshort$ allows elements of the first set to be matched with larger elements. We first prove an auxiliary lemma.

\begin{lemma}\label{lem:INTDM_prop}
    Any $\INTDMshort$ instance $I=(A,B,T)$ satisfying the properties of Lemma~\ref{lem:range} can be reduced in polynomial time to an \INTDMshort\ instance $I'=(A',B,T')$ such that:
    
    \begin{enumerate}
        \item $\min(A')=n$.
        \item $\max(A')<3n$.
        \item $\max(T')<4n$.
        \item $T'$ contains distinct elements.
    \end{enumerate}
\end{lemma}

\begin{proof}
     Let $I=(A,B,T)$ be an $\INTDMshort$ instance satisfying the properties of Lemma~\ref{lem:range}. Observe that, if $\max(T) > \max(A) + n$, then $\max(T)$ cannot be satisfied by any $(a,b)\in A\times B$, rendering~$I$ a trivial no-instance. Hence, for our reduction we can assume
    \begin{equation}\label{eq:INTDM}
        \max(T) \leq \max(A) + n.
    \end{equation}

    We now transform $I$ as follows to obtain an equivalent $\INTDMshort$ instance $I'=(A',B,T')$, while preserving the properties of Lemma~\ref{lem:range}.
    
    \begin{itemize}
        \item If $\min(A) < n$, increase all $a \in A$ and all $t \in T$ by $n - \min(A)$.
        \item If $\min(A) > n$, decrease all $a \in A$ and all $t \in T$ by $\min(A) - n$. Note that no element in $T$ will become non-positive by this operation, since $t > \min(A)$, $\forall t \in T$, by Lemma~\ref{lem:range}.
    \end{itemize}
    
    We obtain sets $A',\ T'$ with $\min(A')=n$ and $\max(A')-\min(A')\leq 2n-2$, implying $\max(A')< 3n$. At this point, we have proven the first two desired properties for $I'$.
    For the third property, we use~(\ref{eq:INTDM}) to obtain $\max(T') \leq \max(A') + n < 3n +n = 4n$. By construction, $T'$ contains distinct elements, thus proving all four desired properties.

    Regarding the correctness of the reduction: since we added the same value to all $a \in A$ and all $t \in T$, any triplet $(a,b,t)\in A\times B\times T$ that satisfied $a+b \geq t$ will satisfy the respective inequality with the modified elements and vice versa, implying that $I$ and $I'$ are equivalent. 
\end{proof}

\begin{theorem}\label{thrm:reduction_2}
    $\INTDMshort$ reduces to $\PM$ in polynomial time.
\end{theorem}

\begin{proof}
    By Lemmas~\ref{lem:range},~\ref{lem:INTDM_prop}, we may consider an $\INTDMshort$ instance $I=(A,B,T)$, $B=\{1,\ldots,n\}$, satisfying the four properties of Lemma~\ref{lem:INTDM_prop}.  
    We construct a \PM\ instance $I^*=(A^*,B^*,T^*)$ based on $I$ as follows.

    \begin{itemize}
        \item $A^*$ is a non-decreasing sequence consisting of the elements in $A$ and $3n$ copies of $4n$.
        \item $B^*$ is an increasing sequence consisting of the elements in $B$ and elements $n+1,\ldots,4n$.
        \item $T^*$ is a set containing the elements in $T$, along with elements $5n+1, \ldots, 8n$.
    \end{itemize}
    
    $A^*, B^*, T^*$ have size $4n$ each. We call the $3n$ new elements added to each set \emph{dummy} elements. Observe that all dummy elements are larger than all respective non-dummy elements, by properties~$2,3$ of Lemma~\ref{lem:INTDM_prop}.
    Note that $T^*$ only contains distinct positive integers, by properties~$3,4$ of Lemma~\ref{lem:INTDM_prop}. Hence, in order to prove that $I^*$ is a valid $\PM$ instance, it suffices to prove that $B^*$ is the discretized sequence of $A^*$. 
    
    $A^*$ contains $3n$ copies of $4n$, implying that the $3n$ largest elements of its discretized sequence are $n+1,\ldots,4n$. For all non-dummy elements $a\in A^*$ it holds that $n \leq a < 3n$, by properties~$1,2$ of Lemma~\ref{lem:INTDM_prop}; hence, we obtain that the discretized sequence of $A^*$ is $\langle 1,\ldots,4n \rangle$, which is exactly equal to $B^*$. This proves that $I^*$ is a valid $\PM$ instance. We will now prove that $I^*$ has a solution if and only if $I$ has a solution.

    For any solution of $I^*$, the target $8n$ can only be satisfied by matching one of the dummies in $A^*$ (of value $4n$) with the largest dummy in $B^*$. Suppose we remove the respective triplet from the input. Inductively, the largest remaining dummy target in $T^*$ can only be satisfied by matching a dummy of $A^*$ (of value $4n$) with the largest available dummy in $B^*$. Thus, for any feasible solution of $I^*$, all dummies of $A^*$ must be matched with dummies of $B^*$, satisfying all dummy targets $t\in T^*$.

    We infer that $I^*$ has a solution if and only if there is a partial solution for its non-dummy elements, i.e., the non-dummy elements $a\in A^*, b \in B^*, t\in T^*$ can be matched in triplets $(a,b,t)$ s.t. $a\geq b$ and $a+b \geq t$. However, the first of these two inequalities holds for all non-dummy $(a,b)$ pairs, since $a \geq n$ for all non-dummy $a \in A^*$ and $b \leq n$ for all non-dummy $b \in B^*$. This implies that any solution of $I$ is a feasible partial solution for the non-dummy elements in $I^*$ (and vice versa, with the other direction following immediately from the definitions of the two problems). From all the above, we obtain that $I$ has a solution if and only if $I^*$ has a solution. The reduction described is polynomial-time, since only $3n$ elements were added to each set.\footnote{This is the reason why we had to bound the maximum value of $A$ through Lemma~\ref{lem:range}; otherwise, the reduction would need to add an exponential amount of elements.}
\end{proof}

\begin{corollary}\label{cor:consecutive}
    $\PM$ is NP-complete even when the discretized sequence consists of consecutive numbers and all input numbers are bounded by $\bigO(n)$.
\end{corollary}

\subsection{Reducing Position Matching to 2-Visits}

The following observation follows immediately from Def.~\ref{def:SM} and is an important tool for the reductions presented below.

\begin{observation}\label{obs:sum_matching_increment}
    Let $(D,A,T)$ be a $\PM$ instance and $c$ be any positive integer. Suppose we increase every $e \in D \cup A$ by $c$ and every $t \in T$ by $2c$, obtaining modified sets $D',A',T'$. Recall that $A$ is the discretized sequence of $D$ (see Def.~\ref{def:SM}). Observe that $A'$ is the discretized sequence of $D'$ and thus $(D',A',T')$ is a valid $\PM$ instance. $(D',A',T')$ is a yes-instance of $\PM$ if and only if $(D,A,T)$ is a yes-instance of $\PM$.
\end{observation}

\begin{lemma}\label{lem:SM_red}
    Any $\PM$ instance $I=(D,A,T)$ can be reduced to a \PM\ instance $I'=(D',A',T')$ such that $t>a$, $\forall t \in T', a\in A'$, in polynomial time.
\end{lemma}

\begin{proof}
    We use Observation~\ref{obs:sum_matching_increment} to add $a_n$ to each $e\in D\cup A$ and $2a_n$ to each $t \in T$, obtaining an equivalent $\PM$ instance $I'=(D',A',T')$. We have $\max(A')=2a_n < \min(T')$, which proves the lemma.
\end{proof}

\begin{observation}\label{obs:consecutive}
    If Lemma~\ref{lem:SM_red} is applied to a $\PM$ instance $I=(D,A,T)$ such that $A$ consists of consecutive numbers, then $A'$ will also consist of consecutive numbers.
\end{observation}

We now present our main theorem, whose proof utilizes results from Section~\ref{sec:two_visits_algo}.

%\textcolor{red}{justify polynomial size of S and L}

\begin{theorem}[Main Theorem]\label{thrm:reduction_3}
    $\PM$ reduces to $\twoV$ in polynomial time.
\end{theorem}

\begin{proof}
    We begin by restricting the $\PM$ problem to specific instances.
    
    By Corollary~\ref{cor:consecutive}, Lemma~\ref{lem:SM_red} and Observation~\ref{obs:consecutive}, we may assume a $\PM$ instance $I=(D=\langle d_1,\ldots,d_n\rangle ,\ A=\langle a_1,\ldots,a_n\rangle ,\ T=\{t_1,\ldots,t_n\})$ such that $A$ is the discretized sequence of $D$ and consists of consecutive numbers and $t > a$, $\forall t \in T, a\in A$. Additionally, $t_i \neq t_j$, $\forall i\neq j$ $(i,j \in [n])$, by Def.~\ref{def:SM}. We assume $D,A,T$ are all sorted in non-decreasing order. 
    If $a_1$ is even, we use Observation~\ref{obs:sum_matching_increment} to add $1$ to each $e\in D\cup A$ and $2$ to each $t \in T$, obtaining an equivalent $\PM$ instance. From now on, we assume that $a_1$ is odd.

    By Def.~\ref{def:disc}, it holds that $d_n=a_n$, which implies that the largest possible sum (of an element in $D$ and an element in $A$) is $2a_n$. If $t_n > 2a_n$, then $I$ is a trivial no-instance, since $t_n$ is impossible to satisfy. From now on, we assume that $t_n \leq 2a_n$.

    We will now build an instance $D'$ of $\twoV$ from an instance $I=(D,A,T)$ of \PM\ with the aforementioned assumptions. We let $D'= S || D || L$,\footnote{'$||$' denotes sequence concatenation.} where:
    
    \begin{itemize}
        \item $S$ contains \emph{small} deadlines $s_i = 2i - 1$, $i \in [(a_1-1)/2]$. Note that these are all the odd integers smaller than $a_1$.
        \item $L$ contains \emph{large} deadlines $l_i$, one for each number $l \in [a_n+1, 2a_n]\setminus T$. Since $t_n \leq 2a_n$, it holds that $T \subseteq [a_n+1, 2a_n]$, so every number in $[a_n+1, 2a_n]$ is either in $L$ or in $T$.
    \end{itemize}

    Note that $|S|$ and $|L|$ are bounded by $\bigO(n)$ through Corollary~\ref{cor:consecutive}.
    We will prove that $I$ is equivalent to $D'$ by proving that all visits of nodes $s_i \in S$ and $l_i \in L$ can be placed feasibly in a $\twoV$ schedule satisfying the properties of Theorem~\ref{theorem:second_visits}, regardless of the placement of the visits of nodes $d_i\in D$. For convenience, we will say ``nodes $s_i$'' instead of ``nodes with deadlines $s_i$'' throughout this reduction (same for $d_i$, $l_i$).

    Let $A'$ be the discretized sequence of $D'$. Observe that all $s_i \in S$ and $l_i \in L$ are distinct and they are smaller or larger than all $a_i \in A$ respectively. Hence, $A'$ consists of all $s_i\in S$, all $a_i \in A$ and all $l_i\in L$, by Def.~\ref{def:disc}. Since all $s_i\in S$ are distinct odd numbers smaller than $a_1$, each $s_i\in S$ forms a cluster by itself and all even-numbered positions up to $a_1$ are gaps. By Theorem~\ref{theorem:second_visits}, we know that the $\twoV$ instance $D'$ has a feasible schedule if and only if it has a feasible schedule $F$ satisfying the properties mentioned in that theorem. By Lemmas~\ref{lemma:clusters},~\ref{lemma:secondary_sorted_by_cluster}, we know that the following hold for $F$:

    \begin{enumerate}
        \item The primary visit of node $s_i$ is placed in position $s_i$, for all $s_i \in S$.
        \item The secondary visit of node $s_i$ is placed in the first available gap, i.e., in position $s_i+1$, for all $s_i \in S$.
    \end{enumerate}

    The above satisfies both the deadlines and the induced deadlines of nodes $s_i\in S$. Since we have assumed $a_1,\ldots,a_n$ are consecutive numbers, they all belong to the same cluster $C$ of $A'$. $C$ also contains all $l_i\in L$ s.t. $l_i < t_1$, by definition. By Theorem~\ref{theorem:second_visits} and Lemma~\ref{lemma:clusters}, $F$ must place the primary visits of nodes $d_i\in D$ and $l_i < t_1$ collectively in a permutation of the positions $a_i\in A$ and $l_i < t_1$. However, since $l_i > a_n = d_n$, $\forall l_i\in L$, the primary visits of nodes $d_i\in D$ cannot be placed in positions $l_i < t_1$. Thus, primary visits of nodes $d_i\in D$ have to be placed in a permutation of the positions $a_i\in A$ and primary visits of nodes $l_i < t_1$ have to be placed in a permutation of the positions $l_i < t_1$.

    We place the primary visit of each node $l_i\in L$ in position $l_i$. Thus, if $L$ is nonempty, the induced deadline of the node $l_1=\min(L)$ is $2l_1 > 2a_n \geq t_n$. This implies that the secondary visit of the node $l_1$ can be placed in position $2a_n +1$ of the schedule; this position is larger than all deadlines in $D'$, hence it cannot be occupied by any primary visit. Since $l_i\in L$ are distinct by definition, the induced deadline of $l_{i+1}$ is larger than the induced deadline of $l_i$ by at least $2$, which implies that the secondary visit of $l_{i+1}$ can be placed directly after the secondary visit of $l_i$, for all $i\in [|L|-1]$. Inductively, this process places all secondary visits of nodes $l_i\in L$ in feasible positions.

    At this point, we have shown that there is a feasible placement of all primary and secondary visits of nodes $s_i\in S$ and $l_i\in L$, satisfying the properties of $F$ imposed by Theorem~\ref{theorem:second_visits} and Lemmas~\ref{lemma:clusters},~\ref{lemma:secondary_sorted_by_cluster}. $F$ is further required to place the primary visits of nodes $d_i\in D$ in a permutation of the positions $a_i\in A$ and the respective secondary visits in the first $n$ gaps that have not been used by preceding clusters (i.e., by the secondary visits of nodes $s_i\in S$), in order of non-decreasing induced deadline. Observe that these $n$ gaps are still available; they are exactly the $n$ positions\footnote{This is why we insist on $T$ being a set and not a multiset throughout our reductions; we need $n$ distinct positions.} $t_i \in T\subseteq [a_n+1, 2a_n]$, which have not been used by nodes $l_i\in L$. This is because $l_i \notin T$ by definition and the primary visits of nodes $l_i\in L$ are placed in positions $l_i$; also, the secondary visits of $l_i\in L$ are placed in positions greater than $2a_n$, which are strictly greater than every $t\in T$.

    We are now ready to show equivalence between $I$ and $D'$. By the above, it is shown that a schedule $F$ with the properties of Theorem~\ref{theorem:second_visits} exists for $D'$ if and only if the primary visits of nodes $d_i \in D$ can be placed in a permutation of the positions $a_i\in A$, s.t. there is some sufficiently large induced deadline for each gap $t_i \in T$ and each $d_i$ is placed in a position smaller than or equal to itself. This corresponds exactly to the $\PM$ instance $I=(D,A,T)$. By Theorem~\ref{theorem:second_visits}, the existence of $F$ is equivalent to the feasibility of $D'$. Thus, $I$ has a solution if and only if $D'$ has a feasible schedule, which completes the reduction.    
\end{proof}

\begin{corollary}[Main Result]\label{cor:main}
    $\twoV$ is strongly NP-complete.
\end{corollary}

\begin{remark}
    $\twoV$ is NP-complete only if multisets are allowed as input. For simple sets, it can be solved in linear time, by Theorem~\ref{thrm:distinct_linear}. Note that, indeed, our reduction from \INTDMshort\ to \PM\ (Theorem~\ref{thrm:reduction_2}) pads the input with multiple copies of the number $4n$.
\end{remark}

\section{Hardness for variants of k-Visits and Pinwheel Scheduling}\label{sec:hardness_extra}

In this section we transfer the NP-hardness of \twoV\ to variations of \kV\ and \PWS.

\subsection{k-Visits with varying deadlines}

We start with a generalization of \kV\ in which the deadline of each node is not constant throughout the schedule, but \emph{varies} depending on how many times that node has been visited.

\begin{definition}[$\varkV$]\label{def:var_kV}
    Given $nk$ positive integers $d_{ij}$, $i\in [n],\ j\in [k]$, the \emph{Variable} $\kV$ problem $(\varkV)$ asks whether there exists a schedule of length $nk$, containing each $i \in [n]$ exactly $k$ times, with the constraint that the $j$-th occurrence of $i$ is at most $d_{ij}$ positions away from the previous one (or $d_{i1}$ positions from the beginning of the schedule, if $j=1$).
\end{definition}

\begin{theorem}\label{thrm:hardness_generalized}
    $\varkV$ is strongly NP-complete for all $k\geq 2$, even when $d_{i1}=d_{i2}$ and $d_{ij}=d_{ij'}$, $\forall j,\, j'\in [3,k]$, for all nodes $i\in [n]$.
\end{theorem}

\begin{proof}
    We will reduce $\twoV$ to $\varkV$ ($k\geq 2$). Let $D=\langle d_1,\ldots,d_n\rangle $ be a $\twoV$ instance. We build a $\varkV$ instance $G$ with deadlines $g_{ij}$, $i\in [n]$, $j\in [k]$. We define
    \[g_{ij}=\begin{cases}
        d_i,\ j \leq 2\\
        3n,\ j > 2
    \end{cases}.\]

    We now prove that $D$ and $G$ are equivalent.
    
    ($\Rightarrow$) If there is a feasible schedule for $D$, then we can extend it by visiting all nodes $1,\ldots,n$ in this order $k-2$ times in a row, after its end. Since $g_{ij}=3n$ for $j>2$, the modified schedule is feasible for $G$.

    ($\Leftarrow$) Suppose there is a feasible schedule $F$ for $G$. We will create another schedule $F'$ based on $F$ in the following manner, and prove that it is feasible for $D$. Scan $F$ from start to end. When a first or second visit of any node is encountered, add a visit of that node to $F'$ in its earliest unoccupied position. Thus, the $2n$ positions of $F'$ contain exactly two visits of each node, in the same order in which they were encountered in $F$. The following properties hold for every $i\in [n]$:
    \begin{itemize}
        \item The first visit of $i$ in $F'$ occurs no later than in $F$.
        \item The second visit of $i$ in $F'$ is at least as close to its respective first visit as it was in $F$.
    \end{itemize}
    Hence, all first and second visits have been placed feasibly in $F'$, which implies that $F'$ is a feasible schedule for $D$.

    We infer that $D$ has a feasible schedule if and only if $G$ has a feasible schedule, which completes the reduction.
\end{proof}

\subsection{Pinwheel Scheduling with varying deadlines}

    Similarly, we define a generalization of $\PWS$ in which the deadline of each node changes once after a given threshold of visits.

\begin{definition}[$\TPWS$]
    Given $2n$ positive integers (deadlines) $d_{ij}$, $i\in [n],\ j \in \{1,2\}$, and $n$ positive integers (thresholds) $t_i$, $i\in [n]$, the \TPWS\ problem asks whether there exists an infinite schedule $p_1, p_2, \ldots$ , where $p_j \in [n]$ for $j\in \mathbb{N}$, such that for all $i\in [n]$:
    \begin{itemize}
        \item Up until the $t_i$-th occurrence of $i$, any $d_{i1}$ consecutive entries contain at least one occurrence of~$i$.
        \item After the $t_i$-th occurrence of $i$, any $d_{i2}$ consecutive entries contain at least one occurrence of~$i$.
    \end{itemize}
\end{definition}

\begin{theorem}\label{thrm:hardness_TPWS}
    $\TPWS$ is strongly NP-hard, even when $t_i = t_j,\, \forall i,j\in [n]$ (even with explicit input).
\end{theorem}

\begin{proof}
    The proof is a straightforward modification of the proof of Theorem~\ref{thrm:hardness_generalized}. We reduce \twoV\ to \TPWS. 

    Let $D=\langle d_1,\ldots,d_n\rangle $ be a $\twoV$ instance. We construct a \TPWS\ instance $G$ as follows. We use thresholds $t_i=2,\, \forall i\in [n]$, and deadlines

    \[g_{ij}=\begin{cases}
        d_i,\ j = 1\\
        3n,\ j = 2
    \end{cases}.\]

    We now prove that $D$ and $G$ are equivalent.
    
    ($\Rightarrow$) If there is a feasible schedule for $D$, then we can extend it by periodically visiting all nodes $1,\ldots,n$ in this order, after its end. It follows that this modified schedule is feasible for $G$, since $t_i=2,\, g_{i2}=3n,\, \forall i\in [n]$.

    ($\Leftarrow$) Suppose there is a feasible (infinite) schedule for $G$. Then, one can create a feasible (finite) schedule for $D$ by removing all visits from $G$ except the first two visits to each node. This preserves feasibility (cf. the proof of Theorem~\ref{thrm:hardness_generalized}).

    Thus, $D$ has a feasible schedule if and only if $G$ has a feasible schedule, which concludes the reduction.
\end{proof}

%The strong NP-hardness result of Theorem~\ref{thrm:hardness_TPWS} carries over to an even more general version of \PWS, in which each node changes its deadline a finite amount of times (instead of just once).

We remark that membership in NP for \TPWS\ is open. On the other hand, it could be the case that the problem is PSPACE-complete. Membership in PSPACE for \TPWS\ can be proven as follows, by modifying the respective proof for \PWS~\cite{Holte_Pinwheel}.

\begin{theorem}
    \TPWS\ is in $\mathrm{PSPACE}$.
\end{theorem}

\begin{proof}
    Let $I=(\{d_{ij}\},\, \{t_i\}, i\in [n], j\in \{1,2\})$ be a \TPWS\ instance.
    For $i\in [n]$ define $f(i)=d_{i1}\cdot t_i$ and let $i^* \in [n]$ be the node that maximizes $f(i)$. Observe that, for any feasible schedule of $I$, each node $i\in [n]$ has to be visited at least $t_i$ times up until time slot $f(i^*)$; intuitively, this means that the deadlines of all nodes will have switched to their second value after time slot $f(i^*)$, in any feasible schedule.

    For every time slot $p\in \mathbb{N}$ of an infinite schedule $S$, we define a \emph{state vector} $V(p)\in \mathbb{N}^{n}$, with the $i$-th element of $V(p)$ representing the remaining time before the deadline of node~$i$ expires (after time slot $p$). More formally, define $V(0)=[d_{11},\ldots,d_{n1}]$ and then, for all integers $p>0$, obtain $V(p)$ from $V(p-1)$ as follows: decrease all elements of $V(p-1)$ by $1$, except for the one corresponding to the node $i$ visited in the $p$-th time slot of $S$, which is instead reset to~$d_{i1}$ if $i$ has not yet been visited $t_i$ times, or to $d_{i2}$ otherwise. Observe that $S$ is feasible if and only if all vectors $V(p),\, p\in \mathbb{N}$, contain only strictly positive elements (informally: no deadline ever expires).

    Let $m = \prod_{i\in [n]} d_{i2}$ be the product of the \emph{second} deadlines in $I$. We will now prove that it suffices to only check the existence of a schedule respecting the deadlines up to a certain time slot (without needing to consider any later time slot).

    \begin{claim}\label{claim:TPWS}
        A schedule such that $V(p)$ contains only strictly positive elements for all $p\in [f(i^*)+m+1]$ exists if and only if an infinite schedule exists.
    \end{claim}

    \begin{claimproof}
        We only need to prove the forward direction, as the converse is trivial (by definition of the state vector).
        Let $S$ be a schedule such that all vectors $V(p)$, $p\in [f(i^*)+m+1]$, contain only strictly positive elements.
        Since, as explained above, all nodes have switched to their second deadline after $f(i^*)$ time slots, there are at most $m$ distinct state vectors (with strictly positive elements) for time slots after $f(i^*)$. Hence, by the pigeonhole principle, there exist time slots $p,\, q \in [f(i^*),\, f(i^*)+m+1]$, $p<q$, such that $V(p)=V(q)$ in $S$. By the definition of the state vector it follows that we can obtain a periodic infinite schedule by copying $S$ up to position $q-1$ and then infinitely repeating the sub-schedule of $S$ consisting of time slots from $p$ up to $q-1$.
    \end{claimproof}

    Using Claim~\ref{claim:TPWS}, we can solve \TPWS\ in polynomial space as follows. We produce all feasible schedules by non-deterministically guessing a node visit $j\in [n]$ for each time slot, while only storing the current time slot and state vector, and for each node how many times it has been visited so far. Note that this information suffices to compute the state vector of the next time slot, without needing to store the schedule explicitly. By Claim~\ref{claim:TPWS}, it suffices to check each schedule only up to time slot $f(i^*)+m+1$.\footnote{Although this number can be exponentially large, its representation requires polynomially many bits.}
    %: $I$ admits an infinite schedule if and only if at least one schedule has state vectors with strictly positive elements up until time slot $f(i^*)+m+1$. 
    This can be done in NPSPACE, and thus, by Savitch's theorem, in PSPACE as well.
\end{proof}

\subsection{An optimization version of Pinwheel Scheduling}

As a final corollary of the hardness of \twoV, we obtain that an optimization version of \PWS\ seeking to maximize the times that each node is visited is strongly NP-hard. Note that the answer to the following problem is $\infty$ for any yes-instance of \PWS; however, the problem could be of interest for no-instances of \PWS.

\begin{definition}[\MVPWS]
    Given a non-decreasing sequence of $n$ positive integers (deadlines) $d_i$, $i\in [n]$, the $\MVPWS$ problem asks to find the maximum number $v \in \mathbb{N}\cup \{\infty\}$, such that there exists a schedule visiting every node $i\in [n]$ at least $v$ times each, with the constraint that every occurrence of $i$ is at most $d_i$ time units away from the previous one (or from the beginning of the schedule, if it is the first one).
\end{definition}

\begin{theorem}
    \MVPWS\ is strongly NP-hard.
\end{theorem}

\begin{proof}
    It is clear that, for all $k\in \mathbb{N}$, a \kV\ schedule exists if and only if the answer $v$ to the respective \MVPWS\ instance is not smaller than $k$. Hence, \twoV\ reduces to \MVPWS\ in polynomial time.
\end{proof}

% \begin{remark}
%     Membership in NP for $\TPWS$ is open. On the other hand, it could be the case that the problem is PSPACE-complete. Membership in PSPACE can be proven as follows.
%     One can go through all the finite schedules up to position $\max_i (d_{i1} \cdot t_i)$ in polynomial space and the remainder of any feasible schedule becomes cyclic after at most $\prod_{i=1}^n d_{i2}$ positions (cf. the PSPACE-membership proof for $\PWS$~\cite{Holte_Pinwheel}).
% \end{remark}

%\textcolor{red}{Mention that optimization version of pinwheel (maximize turns) is strongly NP-hard as a corollary}

\section{Conclusion}
    Our work highlights that finite versions of $\PWS$ are strongly NP-complete, also implying strong NP-hardness for a natural generalization of the infinite version, which we call $\TPWS$. Thus, we show that \textsc{Pinwheel Scheduling} becomes strongly NP-hard when the deadline of each task is allowed to change once during the schedule, which may provide insight towards settling the long-standing open questions concerning the time complexity of the problem. %Our results contribute to the ongoing understanding of the computational complexity of deadline-based scheduling problems.

    Our findings lead to several interesting directions for future work. One such direction is to explore special cases in which $\twoV$ is tractable. We believe there may exist an FPT algorithm parameterized by the \emph{number of numbers} (i.e., the number of distinct deadlines in our context), which is a common parameterization for numerical matching problems~\cite{Number_of_numbers}; we already showed that the problem is tractable when there are up to two distinct numbers. Also, since $\twoV$ is in P for distinct deadlines, there may exist an FPT algorithm parameterized by some parameter related to the multiplicity of the input elements. Note, however, that further work has proven that \twoV\ is strongly NP-complete even for maximum multiplicity $2$~\cite{ICALP_kVisits}.   
    Another natural direction for future research is to study the complexity of the $k$-\textsc{Visits} problem for $k>2$. We conjecture that $3$-\textsc{Visits} is strongly NP-complete even for distinct deadlines. %Should this conjecture be proven, it would significantly deepen our understanding of the complexity of deadline--based scheduling problems.

    The $\kV$ problem that we studied here can be modeled by an agent traversing a complete graph with deadlines. Equivalently, the problem can be modeled by a star graph where the center is a special node with no deadline. Hence, a natural direction for future work is to study the $\kV$ problem in other graph classes, e.g. lines or comb graphs. This could help identify additional tractable instances for deadline-based problems.

    Our results and the open questions arising from this work may contribute to the ongoing effort to understand the complexity of $\PWS$. It still remains open whether $\PWS$ is strongly NP-hard,\footnote{We remark that a very recent preprint by Kobayashi, Lin, and Swernofsky~\cite{kobayashi2026densepinwheelpackingstrongly} claims to have answered this question.} whether it is contained in NP, and whether it is PSPACE-complete. By analyzing $\kV$ and $\TPWS$, we provide new insights that could help address these questions.

\subsubsection*{Acknowledgements}

We are grateful to Shantanu Das for valuable discussions and insights regarding the \kV\ problem. We would also like to thank the anonymous reviewer in SODA 2026 who suggested a concise proof for Lemma~\ref{lem:SM_small_target}, as well as the anonymous reviewers of the TALG journal for assisting us in correcting minor errors and for suggesting an improvement in the running time of the algorithm of Section~\ref{subsec:cluster}.

This work has been partially supported by project MIS 5154714 of the National Recovery and Resilience Plan Greece 2.0 funded by the European Union under the NextGenerationEU Program. Sotiris Kanellopoulos and Christos Pergaminelis acknowledge the financial support of the Archimedes Research Unit (Athena RC) during their PhD.
Maria Kokkou was supported by the DFG Project SCHE 1592/10-1.

No artificial intelligence was used at any point during the production of this paper. All of the presented results were conceived exclusively by the authors and the entirety of the paper was written exclusively by the authors.

%########################################################
%Citation Inventory (for later organization of intro)

%the originals
%\nocite{Holte_Pinwheel}
%\nocite{Holte_2_distinct}
%\nocite{Lin_3distinct}
%\nocite{Chan_conjecture} %stated the conjecture
% the 5/6 conjecture
%\nocite{Kawamura_5/6_stoc}
%\nocite{Gasieniec_towards_5/6}
%\nocite{Kawamura_pinwheel_cover}

%other versions of pinwheel
%\nocite{Feinberg_Generalized_Pinwheel}
%\nocite{Biktairov_Polyamorous}
%\nocite{Korst_scheduling} % periodic task scheduling

%on the complexity of pinwheel 
%\nocite{Jacobs_Window_Scheduling_Complexity} %arxiv
%\nocite{Bosman_Replenishment}

%bamboo
%\nocite{Bamboo_approx_1}
%\nocite{Bamboo_approx_2}
%\nocite{Bamboo_first}
%\nocite{Bamboo_second}

%matching
%\nocite{Flow_shop_Yu}
%\nocite{Number_of_numbers} %no reason to cite this, after all

\bibliographystyle{splncs04}
\bibliography{bibliography}
%########################################

\end{document}